\PassOptionsToPackage{unicode}{hyperref}
\PassOptionsToPackage{hyphens}{url}
\PassOptionsToPackage{dvipsnames,svgnames,x11names}{xcolor}
\documentclass[12pt]{article}

\usepackage{amsmath,amssymb,amsthm}
\usepackage{iftex}
\ifPDFTeX
  \usepackage[T1]{fontenc}
  \usepackage[utf8]{inputenc}
  \usepackage{textcomp} 
\else 
  \usepackage{unicode-math}
  \defaultfontfeatures{Scale=MatchLowercase}
  \defaultfontfeatures[\rmfamily]{Ligatures=TeX,Scale=1}
\fi
\usepackage{lmodern}
\ifPDFTeX\else
\fi
\IfFileExists{upquote.sty}{\usepackage{upquote}}{}
\IfFileExists{microtype.sty}{
  \usepackage[]{microtype}
  \UseMicrotypeSet[protrusion]{basicmath} 
}{}
\makeatletter
\@ifundefined{KOMAClassName}{
  \IfFileExists{parskip.sty}{%
    \usepackage{parskip}
  }{
    \setlength{\parindent}{0pt}
    \setlength{\parskip}{6pt plus 2pt minus 1pt}}
}{
  \KOMAoptions{parskip=half}}
\makeatother
\usepackage{xcolor}
\makeatletter
\ifx\paragraph\undefined\else
  \let\oldparagraph\paragraph
  \renewcommand{\paragraph}{
    \@ifstar
      \xxxParagraphStar
      \xxxParagraphNoStar
  }
  \newcommand{\xxxParagraphStar}[1]{\oldparagraph*{#1}\mbox{}}
  \newcommand{\xxxParagraphNoStar}[1]{\oldparagraph{#1}\mbox{}}
\fi
\ifx\subparagraph\undefined\else
  \let\oldsubparagraph\subparagraph
  \renewcommand{\subparagraph}{
    \@ifstar
      \xxxSubParagraphStar
      \xxxSubParagraphNoStar
  }
  \newcommand{\xxxSubParagraphStar}[1]{\oldsubparagraph*{#1}\mbox{}}
  \newcommand{\xxxSubParagraphNoStar}[1]{\oldsubparagraph{#1}\mbox{}}
\fi
\makeatother

\usepackage{longtable,booktabs,array}
\usepackage{calc} 
\usepackage{etoolbox}
\makeatletter
\patchcmd\longtable{\par}{\if@noskipsec\mbox{}\fi\par}{}{}
\makeatother
\IfFileExists{footnotehyper.sty}{\usepackage{footnotehyper}}{\usepackage{footnote}}
\makesavenoteenv{longtable}
\usepackage[pdftex]{graphicx}
\makeatletter
\def\maxwidth{\ifdim\Gin@nat@width>\linewidth\linewidth\else\Gin@nat@width\fi}
\def\maxheight{\ifdim\Gin@nat@height>\textheight\textheight\else\Gin@nat@height\fi}
\makeatother
\setkeys{Gin}{width=\maxwidth,height=\maxheight,keepaspectratio}
\makeatletter
\def\fps@figure{htbp}
\makeatother

\makeatletter
\@ifpackageloaded{caption}{}{\usepackage{caption}}
\AtBeginDocument{%
\ifdefined\contentsname
  \renewcommand*\contentsname{Table of contents}
\else
  \newcommand\contentsname{Table of contents}
\fi
\ifdefined\listfigurename
  \renewcommand*\listfigurename{List of Figures}
\else
  \newcommand\listfigurename{List of Figures}
\fi
\ifdefined\listtablename
  \renewcommand*\listtablename{List of Tables}
\else
  \newcommand\listtablename{List of Tables}
\fi
\ifdefined\figurename
  \renewcommand*\figurename{Figure}
\else
  \newcommand\figurename{Figure}
\fi
\ifdefined\tablename
  \renewcommand*\tablename{Table}
\else
  \newcommand\tablename{Table}
\fi
}
\@ifpackageloaded{float}{}{\usepackage{float}}
\floatstyle{ruled}
\@ifundefined{c@chapter}{\newfloat{codelisting}{h}{lop}}{\newfloat{codelisting}{h}{lop}[chapter]}
\floatname{codelisting}{Listing}

\makeatother
\makeatletter
\@ifpackageloaded{caption}{}{\usepackage{caption}}
\@ifpackageloaded{subcaption}{}{\usepackage{subcaption}}
\makeatother

\ifLuaTeX
  \usepackage{selnolig}  
\fi
\usepackage[compress]{natbib}
\usepackage{bookmark}

\IfFileExists{xurl.sty}{\usepackage{xurl}}{} 
\hypersetup{
  pdftitle={Pivotal and identification-robust  nonparametric
  inference in linear IV models},
  pdfauthor={Bertille Antoine; Pascal Lavergne},
  pdfkeywords={Weak Identification, Specification Testing, Subvector Inference},
  colorlinks=true,
  linkcolor={blue},
  filecolor={Maroon},
  citecolor={Blue},
  urlcolor={Blue},
  pdfcreator={LaTeX via pandoc}}

\newcommand{\anon}{1}

\newcommand{\E}{\operatorname{E}}

\newcommand{\HICM}{\operatorname{HICM}}

\newcommand{\Var}{\operatorname{Var}}

\def\E{\mathbb{E}\,}
\def\Rit{\mathbb{R}}
\newcommand{\ICM}{\operatorname{ICM}}
\DeclareMathOperator{\diag}{diag}
\newcommand{\mI}{\mbox{\bf I}}
\newcommand{\mO}{\mbox{\bf 0}}

\newcommand{\ip}[2]{\langle #1,#2 \rangle}

\newcounter{thm}[section]
\newtheorem{theor}[thm]{Theorem}

\newtheorem{lem}[thm]{Lemma}

\newtheorem{corollary}{Corollary}[section]
\newtheorem{hp}{Assumption}

\newcommand{\cvp}{\mbox{$\stackrel{p}{\longrightarrow}\,$}}
\newcommand{\cvas}{\mbox{$\stackrel{as}{\longrightarrow}\,$}}
\def\ind{\mathbb{I}}

\begin{document}

\def\spacingset#1{\renewcommand{\baselinestretch}%
{#1}\small\normalsize} \spacingset{1}


\if1\anon
{
  \title{\bf Pivotal and identification-robust \\nonparametric
  inference in linear IV models}
  \author{Bertille Antoine\thanks{
    B. Antoine acknowledges funding from The Social Sciences and Humanities Research Council of Canada. P. Lavergne acknowledges funding from the French National Research Agency (ANR) under the Investments for the Future program (Investissements d'Avenir, grant ANR-17-EURE-0010).}\hspace{.2cm}\\
    Department of Economics, Simon Fraser University\\
    and \\
    Pascal Lavergne\\
    Toulouse School of Economics, Universit\'e Toulouse  Capitole}
  \maketitle
} \fi

\if0\anon
{
  \bigskip
  \bigskip
  \bigskip
  \begin{center}
    {\LARGE \bf Pivotal and identification-robust \\ \vspace{.2cm} nonparametric
  inference in linear IV models}
\end{center}
  \medskip
} \fi

\bigskip
\begin{abstract}
We develop new inference procedures for a linear IV model that are
robust to identification strength and heteroskedasticity of unknown
form, and nonparametric with respect to the first-stage
equation.  Our first test is tailored for inference on parameters of
endogenous explanatory variables. Our new statistic modifies that
of \cite{antoine_identification-robust_2023} to directly account for
heteroskedasticity of unknown form.
As a result, it is asymptotically pivotal, so that inference is greatly
facilitated in practice.  We also develop (i) an identification-robust
subvector inference procedure that does not rely on the knowledge of
identification strength for the remaining parameters, and (ii) a pure
specification test.  In both cases, the tests are conservative but
powerful.
We show that our procedures are computationally friendly and
competitive with existing ones in simulations and an
application.
\end{abstract}

\noindent%
\vfill

\newpage
\spacingset{1.8} 

\section{Introduction}
\label{sec:intro}

We consider cross-section data observations and the linear  model popular from micro-econometrics
\begin{equation}
y_{i} = Y'_{2i} \beta + X'_{1i} \gamma +  u_{i} \qquad \E (u_{i} |
X_{1i}, X_{2i}) = 0 \quad i=1, \ldots n
\, ,
\label{eq:model}
\end{equation}
where $Y_{2}$ are endogenous variables, $X_{1}$ are exogenous control
variables, and $X_{2}$ are exogenous instrumental variables.
Over the last 30 years, it has become clear that standard asymptotic
approximations may reflect poorly what is observed even for large samples when there
is weak correlation between instrumental variables and endogenous
explanatory variables. Alternative asymptotic frameworks have
been developed to account for potentially weak identification, along with
tests that deliver reliable inference about
parameters of interest, see e.g., \cite{SS}, \cite{SW},
\cite{Moreira}, \cite{kleibergen_pivotal_2002,Kleib},
\cite{Andrews2012}, \cite{andrews_identification_2019},
\cite{Andrews2016}, and
\cite{andrews_conditional_2016,andrews_geometric_2016}.
 Surveys on weak identification issues include
\cite{SWY}, \cite{Dufour}, \cite{HH}, \cite{AS}, and
\cite{Andrews2019}.
These inference procedures are robust to identification strength and
uniformly control size. They have good power properties when the first stage is linear
and well specified, in particular,  the conditional likelihood ratio
test of \cite{Moreira} is
nearly optimal, see \cite{AMS2006} and \cite{andrews_optimal_2019}.
However,  reliance on a parametric first stage  may be misleading.
From an empirical perspective, \cite{Dieterle2016} have documented
significant nonlinearities in first-stage regression in several
applied microeconomics papers.   Since practitioners typically have little prior
information on the form of the relation between endogenous variables
and instruments, one may consider estimating the reduced form
nonparametrically, e.g., using an increasing  number of approximating
series. However,  nonparametrically estimated  instruments cannot be
relied upon under weak identification, see
\cite{JP11} and \cite{mikusheva_inference}.
Indeed, if identification is not strong enough,
the statistical variability of a nonparametric estimator will dominate
the signal we aim to estimate.

In a recent work, building on the Integrated Conditional Moment (ICM)
 principle  proposed by \cite{Bierens82},
\cite{antoine_identification-robust_2023} develop
two inference procedures that are robust to any identification pattern
and unknown heteroskedasticity, and do not rely on a linear projection
in the first-stage equation. In particular, they study an ICM test
that tests at the same time for the value of the parameter and the
specification of the model.  However, critical values should be
simulated for each value of the parameter.

Our present work builds on and improves upon the ICM test.  We
 propose a new heteroskedastic version of their ICM statistic, labeled
 HICM. Its key feature is that it is asymptotically pivotal, i.e.,
 under the null hypothesis $H_0:
\beta=\beta_{0}$, its asymptotic distribution is independent of
$\beta_0$. Since building a confidence interval for $\beta$ involves
inverting a test statistic for $H_0$ a large number of times,
the key advantage of our test is computational, as using critical
values that are independent of parameter values significantly speeds
up the process. In our implementation, we found that computation time
can be reduced by a factor  close to 70. In practice, this means
that when ICM-based inference can take weeks in some instances, it
would reduce to a few hours with our new procedure.

We also develop two additional inference procedures that build on this
particular feature and cannot be entertained using the ICM test of
\cite{antoine_identification-robust_2023}.  First, we investigate
subvector inference and obtain valid identification-robust inference
on some parameters of interest, regardless of the underlying
identification of the other parameters not under test. Second, we
devise a {\em pure} specification test that is powerful independently
of identification strength and of the particular functional form of
the reduced form that links instruments and endogenous variables.  In
both cases, inference is conservative but easily implemented.  In
addition, the different tests uniformly control size irrespective of
identification strength.

We  contribute to the recent literature  concerned with joint
issues of specification  and weak identification.
With a focus on inference  under potential misspecification,
\cite{KleibergenZhan} propose a
doubly-robust Lagrange multiplier statistic to deliver
identification-robust and valid inference on the pseudo-true value in
a moment-based model.
With a focus on specification testing,
\cite{AntoineFrazierRenault2026WP}  develop procedures
in GMM and minimum distance settings that are robust to weak
identification. \cite{Dovonon2026} propose a  modified J-test with
an increasing number of instruments that is asymptotically pivotal
and agnostic on identification, but requires to choose the number and
identity of nonlinear transformations of the instrumental variables.
The procedures proposed here complement these approaches, while
allowing the first-stage equation to be nonparametric.

We illustrate the finite sample properties of our procedures in a
series of simulations. When the model is correctly specified, we find
that the level of our HICM-based inference is well
controlled, that it has significant power advantages compared to
existing procedures such as the one of \cite{SW} when the reduced form
equation is nonlinear, and that it  is also competitive for a linear
reduced form. Similarly, when testing whether the model is correctly
specified, our HICM-based specification test demonstrates excellent
size and power properties. We also demonstrate the relevance and power
of our procedure in an empirical application that studies the effects
of population decline in Mexico on land concentration in the sixteenth
century, using the data and framework of \cite{Mexico}.

Our paper is organized as follows. In Section \ref{sec.frame}, we
introduce our framework and our new HICM test statistic. We discuss
how it can be used for powerful and identification-robust inference
that is compatible with subvector inference, as well as for
specification testing. The asymptotic properties of our inference
procedures and specification test based on HICM are studied in
Section \ref{sec.asy}. Their finite sample properties are investigated
in a series of Monte Carlo experiments and in an empirical application
in Sections  \ref{sec.simulation} and \ref{sec.ea}. We conclude in Section \ref{sec.concl}.

\section{Framework}\label{sec.frame}

The influence of exogenous control variables $X_{1}$ can be projected
out through orthogonal projection in (\ref{eq:model}), which does not
influence our reasoning but simplifies exposition.  Hence, in what
follows, we consider a structural equation of the form
\begin{equation} \label{eq:linear-SE}
y_i  =  Y_{2i}' \beta + u_{i}  \qquad \E(u_i|Z_i)= 0  \quad i=1, \ldots n \, ,
\end{equation}
augmented by a first-stage reduced form equation for $Y_2$,
\begin{equation} \label{eq:general-RF}
Y_{2i} = \Pi(Z_{i}) + V_{2i} \qquad \E(V_{2i}|Z_i)= 0 \, .
\end{equation}
The variables $Z$ include the instruments $X_{2}$, but may also
include the exogenous $X_{1}$ to account for potential nonlinearities in $X_{1}$
in the function $\Pi(\cdot)$. Combining \eqref{eq:linear-SE} and
\eqref{eq:general-RF} yields
\[
y_i - Y_{2i}'\beta_{0} = \Pi'(Z_i) \left(\beta-\beta_0\right) + \varepsilon_i,
\qquad \text{where} \quad \varepsilon_i = u_i + V_{2i}' \left(\beta - \beta_{0}\right)
\quad \text{and} \quad \E\left(\varepsilon_i|Z_i\right)=0
\, .
\]
We propose to test the null hypothesis
\begin{equation}
\label{eq:H0}
H_{0}: \  \E(y - Y'_{2}\beta_{0}|Z) = 0
\quad \mbox{a.s. }
\, ,
\end{equation}
which considers at the same time the null hypothesis $\beta = \beta_0$ and the correct
specification of the model.
In what follows, we introduce our proposed HICM test, then we discuss
subvector inference and specification testing.

\subsection{Heteroskedasticity-robust ICM statistic} \label{sub:HICM}

\cite{antoine_identification-robust_2023} introduced the ICM test
statistic. Building on  \cite{Bierens82},  ${H}_{0}$  is shown to be equivalent to
\begin{equation}
{H}_{0}: \
\E \left[ \left( y - Y_{2}'\beta_0 \right)
\exp( i s'Z) \right] = 0  \quad \forall s\in\mathbb{R}^k
\, .
\label{icmexp}
\end{equation}
Bierens' Integrated Conditional Moment (ICM) statistic is
\begin{equation}
 \int_{\Rit^{k}}{|n^{-1/2} \sum_{j=1}^{n}{ \left(y_j -
     Y'_{2j}\beta_{0}\right) \exp(i s'Z_{j})}|^{2} \, d\mu (s)}  \, ,
\label{icm}
\end{equation}
where $\mu$ is some symmetric measure with support $\Rit^{k}$.
The  ICM statistic of \cite{antoine_identification-robust_2023} is  a
standardized version of the latter that can be
written as the following quadratic form
\begin{equation} \label{eq:AR-stat}
\ICM (\beta_{0}) = \frac{b'_0 Y' W Y b_{0}}{b_{0}'\widehat{\omega}
b_{0}} \qquad \text{where} \qquad b_0 = (1, - \beta_0')' \quad , \quad
{Y} =
\left[
\begin{array}{cc}
y_{1} & Y'_{21} \\
\vdots & \vdots \\
y_{n} & Y'_{2n}
\end{array}
\right]
\, ,
\end{equation}
 $W$ is the matrix of generic element
$n^{-1}w\left(Z_{j}-Z_{m}\right)$, with $w(z) = \int_{\Rit^{k}} \exp(i
s'z) \, d\mu(s)$, and $\widehat{\omega}$ is a (semiparametric)
estimator of $\omega = \E (\Var(Y|Z))$.  Hence, the ICM statistic
\eqref{eq:AR-stat} resembles the AR statistic, with $W$ replacing
$P_{Z}$, the orthogonal projection on $Z$.  As apparent from its
construction, it is designed to test {the correct specification of the
model} together with the parameter value, as does the AR test under
the assumption of a linear reduced form. While under homoskedasticity,
the statistic has an asymptotic distribution independent of $\beta_0$,
under heteroskedasticity, critical values depend on the considered
$\beta_0$.
A confidence set is generally obtained by inverting the $\ICM$
test as $\left\{ \beta_{0} : \ICM(\beta_{0}) < c_{1-\alpha}(Z, \beta_0)
\right\}$. In particular, this means that, under heteroskedasticity, critical values
$c_{1-\alpha}(Z, \beta_0)$ should be simulated for each candidate
$\beta_0$, see \cite{antoine_identification-robust_2023} for details.
Unless  $\beta_0$ is unidimensional, the associated computational cost can be prohibitive.

We now construct a version of $\ICM$ that directly accounts for unknown
heteroskedasticity in its construction. Let  $\Omega(Z) = \Var(Y|Z)$,
the conditional variance of $Y$ given $Z$. We rewrite our null
hypothesis of interest in the equivalent form
\begin{equation}
{H}_0: \,  \E\left[ (b_0' \Omega(Z) b_0)^{-1/2}  (y - Y'_{2}\beta_{0}) |Z \right] = 0
\quad \mbox{a.s. }
\label{hyp-hicm}
\end{equation}
From \cite{Bierens82}'s results, this is also equivalent to
\begin{equation}
{H}_{0}: \
\E \left[ (b_0' \Omega(Z) b_0)^{-1/2} \left( y - Y_{2}'\beta_0 \right)
\exp( i s'Z) \right] = 0  \quad \forall s\in\mathbb{R}^k
\, .
\label{icmexp2}
\end{equation}
Following a rationale similar to the one above, we can build an ICM
statistic for this hypothesis.
Given a consistent estimator $\widehat{\Omega}(\cdot)$ of $\Omega(\cdot)$,
we thus define
\begin{align}
\HICM (\beta_{0}) & =
n^{-1}  \sum_{j=1}^{n} \sum_{m=1}^{n} {
(b_0' \widehat{\Omega}(Z_j) b_0)^{-1/2} (Y_j' b_0)
(b_0' \widehat{\Omega}(Z_m) b_0)^{-1/2} (Y_m' b_0)
w(Z_j - Z_m) }
\nonumber \\
&  = b_{0}' Y'
\left[ (e \otimes b_0 )'  \widehat{\Omega}  (e \otimes b_0 ) \right]^{-1/2}   W
\left[ (e \otimes b_0 )'  \widehat{\Omega}  (e \otimes b_0 ) \right]^{-1/2}
Y b_{0}
\label{eq:HICM}
\, ,
\end{align}
where $\widehat{\Omega} = \diag\left( \widehat{\Omega}(Z_i), i = 1,
\ldots n\right) $, and $e$ is the $n$-vector of ones.
If the model is correctly specified with parameter $\beta_0$, twe show that $\HICM(\beta_0)$ asymptotically follows the same
distribution as $G' W G$, where $G\sim N (\mO, \mI)$. This
distribution is independent of the particular value of $\beta_0$. We
can then easily simulate the distribution of our statistic under $H_{0}$ and
recover a critical value as the $(1-\alpha)$ quantile of the
distribution of $G' W G$, denoted as $c_{1-\alpha}(Z)$.
A confidence set is obtained by inverting the $\HICM$ test, that is $\left\{ \beta_{0} :
\HICM(\beta_{0}) < c_{1-\alpha}(Z) \right\}$. As we will see,  the
use of critical values that are independent of $\beta_0$ is
computationally very advantageous.

Practically, we need to detail a number of elements of our
statistic. First, one needs to estimate the conditional variance $\Omega(\cdot)$.
One should note that weak identification does not preclude consistent
estimation of this quantity.
The conditional variance can be estimated parametrically if
one is ready to make an assumption on its functional form.  Otherwise,
we can resort to nonparametric conditional variance estimation.
Several consistent estimators have been developed for a univariate
$Y$, and generalize easily.  To make things concrete, we focus on
kernel smoothing, which is used in our simulations and application.
Let
\[
\overline{Y}(z) =  \left(nb_n\right)^{-1} \sum_{i=1}^n Y_i K \left( (
 Z_i - z)/ b_n \right)
\]
based on the $n$ iid observations $(Y_{i},Z_{i})$, a kernel
$K(\cdot)$, and a bandwidth $b_n$.  With $e = (1, \ldots , 1)'$, let
$\widehat{f} (z) = \overline{e} (z)$ and $\widehat{Y} (z) =
\overline{Y}(z)/\widehat{f} (z)$.  The conditional variance estimator
 is defined as
\begin{equation} \label{eq:condit var  estim Y}
  \widehat{\Omega} (z)
  =
\left(nb_n\right)^{-1} \frac{\sum_{i=1}^n \left(Y_{i} - \widehat{Y} (Z_{i})\right)
\left(Y_{i} - \widehat{Y} (Z_{i})\right)'
K \left( ( Z_i - z)/ b_n \right)}{\widehat{f}(z)}
   \, .
\end{equation}
This estimator, studied by \cite{NCM}, is a generalization of the
kernel conditional variance and is positive definite whenever
$K(\cdot)$ is positive.   We can even consider
\[
  \widehat{\Omega} (z)
  =
\left(nb_n\right)^{-1} \frac{\sum_{i=1}^n Y_{i}
Y'_{i} K \left( ( Z_i - z)/ b_n \right)}{\widehat{f}(z)}
   \, ,
\]
which estimates the uncentered moments $\E \left( Y'Y\right)$, and
thus avoids preliminary estimation of $\E (Y|Z)$. Indeed, we only need
$\widehat{\Omega} (z)$ to estimate $\Var(y - Y_2'\beta_0|Z)$,
which equals $\E[(y - Y_2'\beta_0)^2|Z]$ under $H_0$, so estimation
of uncentered cross moments is sufficient for this purpose.

Another element of our statistic is the function $w(\cdot)$, or
equivalently the measure $\mu$. The role of $w(\cdot)$ resembles the
one of the kernel in nonparametric estimation, but, in contrast, it is
a {\em fixed user-chosen function that does not vary with the sample
size}. To make this explicit, we will impose that the squared integral
of $w(\cdot)$ equals one.\footnote{A more involved restriction would
be to impose a similar condition on the Frobenius norm of $W$.}  The
condition for $\mu$ having support $\Rit^{k}$ translates into the
restriction that $w(\cdot)$ should have a strictly positive Fourier
transform almost everywhere.  Examples include products of triangular,
normal, logistic, see \citet[Section 23.3]{JKB95}, Student, including
Cauchy, see \cite{DK2002}, or Laplace densities.  To achieve scale
invariance, we recommend, as in \cite{Bierens82}, to scale the exogenous
instruments by a measure of dispersion, such as their empirical
standard deviation.

If $Z$ has bounded support, results from \cite{Bierens82}  yield
that ${H}_{0}$ holds if and only if
$
\E \left[ (b_0' \Omega(Z) b_0)^{-1/2} \left( y - Y_{2}'\beta_0 \right)
\exp(s'Z) \right] = 0
$
for all $s$ in an arbitrary neighborhood of $0$ in
$\mathbb{R}^q$. Hence $\mu$ can be taken as any symmetric probability
measure that contains $0$ in the interior of its support. For
instance, we can consider the product of uniform distributions on
$[-\pi,\pi]$, so that $w(\cdot)$ is the product of sinc functions.  As
noted by \cite{Bierens82}, there is no loss of generality assuming a
bounded support, as his above-mentioned equivalence result equally
applies to a one-to-one transformation of $Z$, which can be chosen
with a bounded image. Hence, we will adopt this assumption in our
theoretical analysis.

On a more general level, \cite{Bierens82}'s principle replaces
conditional moment restrictions by a continuum of unconditional
moments involving the complex exponential function. Other functions
have been used beyond this choice, see
\cite{Bierens90} and \cite{BP97}. \cite{stinchcombe_white_1998}
characterize  a large class of functions that could
generate an equivalent set of unconditional moments. As detailed by
\cite{LP2013}, this yields a full collection of potential estimators
under strong (or semi-strong) identification, such as the ones developed
by \cite{dominguez_consistent_2004}, \cite{AntoineLavergne}, and
\cite{escanciano_simple_2018} among others. This would also yield a
collection of test statistics that could be used under weak
identification, see \cite{Chen2025} for a recent instance.
Here, we focus on a particular application of the ICM, which is suitable for
theoretical investigation and practical implementation, and we leave
for future work the investigation of the relative merits of these
different ICM-type tests.

\subsection{Subvector inference} \label{subs:subvector}

Let us partition our model \eqref{eq:linear-SE} as
\begin{equation}
y_{i} = Y'_{2i,1} \beta_{0,1} + Y'_{2i,2} \beta_{0,2} +  u_{i} \qquad \E (u_{i} |
Z_{i}) = 0 \quad i=1, \ldots n
\, .
\label{eq:model3}
\end{equation}
Our statistic writes
\[
 \HICM(\beta_{0,1},\beta_{0,2}) = b_{0}' Y'
\left[ (e \otimes b_0 )'  \widehat{\Omega}  (e \otimes b_0 ) \right]^{-1/2}   W
\left[ (e \otimes b_0 )'  \widehat{\Omega}  (e \otimes b_0 ) \right]^{-1/2}
Y b_{0}
\, ,
\]
with $b_0=(1, -\beta_{0,1}',-\beta_{0,2}')$.  Suppose we aim at
identification-robust inference on $\beta_{0,1}$ only, so $\beta_{0,2}$ can
be seen as a nuisance parameter that cannot be easily partialed out or
reliably estimated.

We resort  to a {\em conservative projection approach}, see
  \cite{Dufour1997}. Our test statistic is
\[
\HICM^*(\beta_{0,1})
\equiv \min_{\beta_{2}} \HICM(\beta_{0,1},\beta_2)
\, .
\]
Since $\HICM^*(\beta_{0,1}) \leq \HICM(\beta_{0,1},\beta_{0,2})$, we
can rely on the critical value $c_{1-\alpha}(Z)$ from the null
asymptotic distribution of $\HICM$, which can be easily simulated as
explained above. The approach is conservative but allows to control
the level of the test.  Inverting the test yields the confidence set
$\left\{ \beta_{0,1} : \HICM^*(\beta_{0,1}) < c_{1-\alpha}(Z)
\right\}$ for $\beta_1$.
Valid subvector inference on $\beta_1$ thus obtains without any
assumptions on identification  of $\beta_2$.

\subsection{Specification testing} \label{sub:specif test}

In line with the previous subvector inference procedure, our
specification test is based on
\[
\HICM^* = \min_{\beta} \HICM(\beta)
\, .
\]
Under correct specification, there is a $\beta_0$ such that
(\ref{hyp-hicm}) holds. Since $\HICM^* \leq \HICM(\beta_0)$,
we can rely on the simulated null distribution of $\HICM(\beta_0)$, which
is independent of $\beta_0$, to bound the distribution of $\HICM^*$.
Our asymptotic test rejects the correct specification of the model whenever
$\HICM^* > c_{1-\alpha}(Z)$. From the above inequality, we have
control of asymptotic size, but the test is conservative.
If the model is incorrectly specified, then there is no $\beta_0$ such that
\[
\E \left[ (b_0' \Omega(Z) b_0)^{-1/2}   \left( y - Y_{2}'\beta_0
\right) \exp( i s'Z) \right] = 0
\quad \forall s\in\mathbb{R}^k
\, .
\]
When the above quantity is uniformly bounded away from zero, we expect
$\HICM^*$ to diverge and our test to be consistent. In the next
section, we study the behavior of our specification test under some
local alternatives and show that it has non-trivial power.

\section{Asymptotics}\label{sec.asy}

\subsection{Uniform asymptotic validity}

We consider the following assumptions.
\begin{hp}
\label{ass:iid}
(i) The observations $(y_{i}, Y_{2i}, Z_{i})$ form a rowwise
 independent triangular array, where the  marginal  distribution of
 $Z$   remains unchanged, $Z$ is bounded, and
 for some $\delta > 0$ and $M' < \infty$,
$\sup_{z} \E \left( \|Y\|^{2+\delta} | Z=z\right) \leq M' $  uniformly
 in $n$.
(ii) The observations follow (\ref{eq:linear-SE}) and
 (\ref{eq:general-RF}), and $\beta_0$ belongs to a compact set ${\cal B}$.
\end{hp}
For the sake of simplicity, we do not use a double index for observations, and we
denote by $\left\{Y_{1}, \ldots , Y_{n}\right\}$ the independent copies
from $Y$ for a sample size $n$.
The assumption of a constant distribution for $Z$ could be weakened
but is made to formalize that identification strength is related to
the conditional distribution of $Y$ given $Z$ only.
The boundedness of $Z$ could be relaxed, at the cost of more
technicalities, but, as previously noted, a one-to-one transformation
of $Z$ can be entertained to ensure it holds.

We denote by ${\cal P}$ the class of
distributions on which our observations lie.
Let ${\cal E}$ be a class of vector-valued functions $\Pi(\cdot)$ and let $N \left(
\varepsilon,{\cal E}, L^{2} (Q)\right)$ be the covering number of
${\cal E}$, that is, the minimum number of $L^{2} (Q)$
$\varepsilon$-balls needed to cover ${\cal E}$, where a $L_{2} (Q)$
$\varepsilon$-ball around  $\Pi(\cdot)$ is the set of vector functions
$\left\{ h \in L^{2} (Q) \ : \ \int_{}^{}{ \|h-\Pi\|^{2} \, dQ } < \varepsilon \right\}$.
\begin{hp}
\label{ass:pi}
The conditional expectation vector $\E \left( Y_{2} | Z=\cdot\right)$ belongs
to a class of vector functions  ${\cal E}$  such that
$\forall \, \Pi(\cdot) \in {\cal E}$, $\| \Pi(\cdot) \| \leq F(\cdot)$ with
\[
\lim_{M \rightarrow \infty} \sup_{\cal P}
\E \left[  F^{2} (Z) \ind\left( F (Z) > M \right)  \right] = 0
\]
 and
\[
\log N \left( \varepsilon \E^{1/2} \left( F^{2} (Z)\right), {\cal E}, L^{2} (P) \right) \leq K
\varepsilon^{-V} \quad \mbox{for some } V < 2 \, ,
\]
for all $P \in {\cal P}$ and  some $K,V$  independent of $P$.
\end{hp}
\cite{AndrewsHB} and \cite{van_der_vaart_bracketing_1994}, among
others, exhibit classes of smooth
functions that fulfill the above conditions.

Let ${\cal O}$ be a class of matrix-valued functions, and let $N \left(
\varepsilon,{\cal O}, L_{2} (Q)\right)$ be the covering number of
${\cal O}$, defined similarly as above.
\begin{hp}
\label{ass:variance}
\begin{itemize}
\item[(i)] $\sup_{P \in {\cal P}} \Pr \left[ \| \widehat{\Omega} - \Omega\| > \varepsilon
\right] \rightarrow 0 \quad  \forall \varepsilon >0$.
\item[(ii)] $\Omega (\cdot)$ belongs
to a class of matrix functions  ${\cal O}$  such that
$O < \underline{\lambda} \leq \inf_{z} \lambda_{\min} (\Omega (z)) \leq
\sup_{z} \lambda_{\max} (\Omega (z)) \leq \overline{\lambda} < \infty$
for all $\Omega(\cdot)\in {\cal O}$ and
\[
\log N \left( \varepsilon, {\cal O}, L^{2} (P) \right) \leq K
\varepsilon^{-V} \quad \mbox{for some } V < 2 \, ,
\]
for all $P \in {\cal P}$ and  some $K,V$  independent of $P$.
\item[(iii)] $\sup_{P\in {\cal P}}
\Pr \left( \widehat{\Omega}(\cdot) \in {\cal O}\right) \rightarrow 1$
as  $n\rightarrow\infty$.
\item[(iv)] $ \sup_{P \in {\cal P}}  \int_{}^{}{} \| \widehat{\Omega} (Z) - \Omega (Z)
\|^{2} \, dP (Z) \cvp 0$.
\end{itemize}
\end{hp}
This assumption entails, in particular, that conditional variance
estimation does not affect the asymptotic behavior of our statistics.
There is a tension between the generality of the class of functions
${\cal O}$ and the class of possible distributions ${\cal P}$.
When $\Omega (\cdot)$ is of parametric form, Assumption
\ref{ass:variance} will be satisfied for a large class of
distributions.  When $\Omega(\cdot)$ is considered nonparametric and
estimated accordingly, one typically assumes that its components are
smooth functions, and to prove (iii), one has to show that $\widehat{\Omega}(\cdot)$
also satisfies the same smoothness conditions with probability
converging to 1.  Such results have been derived, see
e.g., \cite{andrews_nonparametric_1995} for kernel estimators or
\cite{cattaneo_optimal_2013} for partitioning estimators.  Uniform
convergence of nonparametric regression estimators (and their
derivatives) generally requires the domain of the functions to be
bounded and the absolutely continuous components of the distributions
of the conditioning variables to have densities bounded away from zero
on their support. When they are not, \cite{andrews_nonparametric_1995}
discusses the use of a vanishing trimming that is compatible with the
stochastic equicontinuity results of \cite{AndrewsHB}. Condition (iv)
is dealt with in the literature on honest confidence intervals using
$L^{2}$ norm, see \cite{robins_adaptive_2006} and the references therein.

\begin{hp}
\label{ass:w}
 $w(\cdot)$ is a symmetric, bounded density  such
 that $\int_{}^{}{w^{2} (x)
 \, dx} = 1$.  Its Fourier transform is a density,
 whose bounded support contains a neighborhood of the origin.
 \end{hp}
We denote by $c_{1-\alpha} (Z)$  the  critical value of $\HICM$
obtained by the simulation-based method detailed above.\footnote{We neglect the
approximation error due to a finite number of simulations by assuming
the number of simulations is infinite so that the critical values are
exact.}
Let ${\cal P}_{\beta_{0}}$ be the subset of distributions in ${\cal
P}$ such that $\beta=\beta_{0}$ in (\ref{eq:linear-SE}).
 The following result establishes that our tests
control size uniformly over a large class of probability distributions.

\begin{theor} Under Assumptions \ref{ass:iid}, \ref{ass:pi},
\ref{ass:variance}, and \ref{ass:w},
\[
\limsup_{n\rightarrow \infty} \sup_{\beta_{0} \in {\cal B}}
\sup_{P \in {\cal P}_{\beta_{0}}}
\Pr\left[ \HICM (\beta_{0}) > c_{1-\alpha}(Z) \right] \leq
\alpha
\, .
\]
\label{th:level}
\end{theor}
Our theorem readily implies that our test is asymptotically valid
whatever the identification strength. Indeed, for any sequence
$\Pi_{n}(\cdot), n \geq 1$, of functions in ${\cal E}$, that can
decrease in norm to zero arbitrarily fast, our result yields
asymptotic validity under this sequence, see  \citet[Chap. 2.8]{vaart_weak_2000}.
The result also readily implies the uniform asymptotic validity of our
subvector inference and
specification tests.
\begin{corollary} Under Assumptions \ref{ass:iid}, \ref{ass:pi},
\ref{ass:variance}, and \ref{ass:w},
\begin{align*}
(i) &
\limsup_{n\rightarrow \infty} \sup_{\beta_{0}  \in {\cal B}}
\sup_{P \in {\cal P}_{\beta_{0}}}
\Pr\left[ \HICM^*(\beta_{0,1}) > c_{1-\alpha}(Z) \right]  \leq
\alpha
\, ,
\\
(ii) &
\limsup_{n\rightarrow \infty} \sup_{\beta_{0}  \in {\cal B}}
\sup_{P \in {\cal P}_{\beta_{0}}}
\Pr\left[ \HICM^* > c_{1-\alpha}(Z) \right]   \leq
\alpha
\, .
\end{align*}
\label{th:levelspec}
\end{corollary}

\subsection{Asymptotic power}

We adopt here a large local alternatives setup similar to
\cite{BP97}.
\begin{hp}
\label{ass:pipower}
There exists a fixed matrix
$C(\cdot)$ such that $\E C(Z) C'(Z)$ is bounded
 and positive definite and either
 (i) $\Pi (Z) =  \tilde{c}_{n} \ \frac{C(Z_{i})}{\sqrt{n}}$, where
 $\tilde{c}_n$ is a sequence of positive real numbers,  or (ii) $\Pi (Z) =  {C(Z_{i})}$.
\end{hp}
 Condition (i) allows studying the power of our tests against weak and
semi-strong identification when considering  a test of $H_{0}: \
\beta = \beta_{1}$ where $\beta_{1} \neq \beta_{0}$, the true parameter value.
Condition (ii) is the strong identification case, and  we
consider local alternatives $H_{1n}: \ \beta_{1n} = \beta_{0} +
\tilde{c}_{n} \frac{\delta}{\sqrt{n}}$, where $\delta \neq 0$ is
fixed.  In both cases, the object of interest is the asymptotic power
of our two tests when $\tilde{c}_{n}$ becomes large.
\begin{theor}
Under  Assumptions \ref{ass:iid},  \ref{ass:variance}, and  \ref{ass:w},
\begin{itemize}
\item[(i)] Under Assumption \ref{ass:pipower}-(i), for  any fixed  $\beta_1\neq\beta_0$,
\[
\liminf_{\tilde{c}_n \rightarrow + \infty}
\inf_{P \in {\cal P}_{\beta_{0}}} \Pr\left[ \HICM (\beta_{1}) >
c_{1-\alpha} (Z) \right]  = 1
\, .
\]
\item[(ii)]
Under Assumption \ref{ass:pipower}-(ii), for  $\beta_{1n}  =
\beta_{0} + \tilde{c}_{n} \frac{\delta}{\sqrt{n}}$ with a fixed $\delta \neq 0$,
\[
\liminf_{\tilde{c}_n \rightarrow + \infty}
\inf_{P \in {\cal P}_{\beta_{0}}} \Pr\left[
\HICM (\beta_{1n}) > c_{1-\alpha}(Z) \right]  = 1
\, .
\]
\end{itemize}
\label{th:power}
\end{theor}
 Result (i) shows that, under weak identification, power is non-trivial
for a large enough $\tilde{c}_{n}$.  Result
(ii) implies that, under strong identification, power is non-trivial under
a sequence of Pitman local alternatives for $\tilde{c}_{n}$ large enough.
A similar statement can be derived for our subvector inference procedure.

\subsection{Asymptotic power of specification test}

We  consider a sequence of local alternatives
\[
H_{1,n}: \
\min_{b = (1, - \beta')' : \beta \in {\cal B}}
\int_{}^{}{ \left| \E \left[ (b' \Omega(Z) b)^{-1/2}  (y -
Y'_{2}\beta) \exp(is'Z)  \right] \right|^2 \, d\mu(s) }  =    \tilde{d}_n^2 / {n}
\, ,
\]
where $ \tilde{d}_n, n = 1, \ldots$ is a positive real sequence uniformly
bounded away from zero. We remain agnostic on the source of
misspecification and on identification strength.
We denote by ${\cal P}$ the class of distributions on which our
observations lie under Assumptions \ref{ass:iid}-(i), \ref{ass:pi} and
\ref{ass:variance}.

\begin{theor}
Under  Assumptions \ref{ass:iid}-(i),  \ref{ass:pi},
\ref{ass:variance}, and  \ref{ass:w},
\[
\liminf_{\tilde{d}_n \rightarrow +\infty} \inf_{P \in {\cal P} \cap H_{1,n}}
\Pr\left[ \HICM^* > c_{1-\alpha} (Z) \right]  = 1
\, .
\]
\label{th:powerspec}
\end{theor}
To understand what our result implies, assume that
\[
\E( y_i | Z_i) = \Pi(Z_i)'\beta_0 + \delta_n(Z_i), \quad i = 1, \ldots n
\, .
\]
This encompasses situations where $Z$ has a direct effect on $y$, and
is thus endogenous, or situations where the relationship between $y$
and $Y_2$ is nonlinear, and thus not appropriately modeled.  If
$\delta_n(Z) = \tilde{d}_n \delta(Z)/\sqrt{n}$ for a non-zero function
$\delta(Z)$, then our specification test has some non-trivial power
for $\tilde{d}_n$ large enough, provided $\delta(\cdot)$ is not
linearly related to the components of $\Pi(\cdot)$, which is implicit
in the formulation of $H_{1,n}$.  This holds whether $\Pi(\cdot)$ is
fixed as in Assumption \ref{ass:pipower}-(ii),  or whether $\Pi(Z) =
\frac{C(Z)}{\sqrt{n}}$  as in Assumption \ref{ass:pipower}-(i).  In that sense, our
specification test is really robust to identification strength.

\section{Small sample behavior}\label{sec.simulation}

We generate data following the model
\begin{eqnarray}
y_i &=& \alpha_0 + Y_{2i}\beta_0 + \delta Z_i^2 + \sigma(Z_i) u_{i} \, ,\label{eq:sim} \\
Y_{2i} &=& \gamma_0 + \frac{c}{\sqrt{n}} f(Z_i) + \sigma(Z_i) v_{2i} \, , \nonumber
\end{eqnarray}
where $c$ is a constant that controls the strength of the
identification and $Y_{2i}$ is univariate.  The joint distribution of
$(u_{i},v_{2i})$ is a bivariate normal with mean $\mathbf{0}$, unit
unconditional variances, and unconditional correlation $\rho$. We set
$\alpha_0=\beta_0=\gamma_0=0$ and $\rho=0.8$.  We consider three
different specifications for the function $f(\cdot)$: (i) a polynomial
function of degree 3 proportional to $z- 2z^3/5$, (ii) a linear
function, and (iii) a function compatible with first-stage group heterogeneity,
see \cite{Abadie2024}, proportional to $\left( 2 z_2 - 1\right) \left( z_1-2z_1^3/5
\right)$. Here $Z$ (or $Z_1$) is deterministic with values
evenly spread between -2 and 2 and $Z_{2}$ follows a Bernoulli with
probability 1/2. Also, $f(Z)$ is centered and scaled to have variance
one to ensure that the different cases are comparable.  We consider
heteroskedasticity depending on the first component of $Z$ of the form
\[
\sigma(z)= \sqrt{\frac{3(1+z^2)}{7}} \, .
\]
Finally, $\delta$ controls the degree of misspecification. When
$\delta = 0$, the model that excludes $Z$ from the structural equation
is well-specified and  similar to the one considered by
\cite{antoine_identification-robust_2023}, while it is misspecified
when $\delta \neq 0$.
In all our experiments, $w(\cdot) = \operatorname{sinc}(\pi \cdot)$, which
corresponds to a uniform density $\mu$, while conditional covariances
are estimated through kernel smoothing
with a Gaussian kernel and rule-of-thumb bandwidth. We consider 5,000
replications.

\subsection{Inference on parameters in a correctly specified model}
\label{subsec:sim idr inference}

For this subsection, we set $\delta=0$ to ensure that the model is
always correctly specified, and we focus on delivering inference on
$\beta_0$. We compare the performance of the HICM test introduced in
this paper and the ICM one.  Simulations results reported in
\cite{antoine_identification-robust_2023} show that among competing
procedures, such as the heteroskedasticity-robust conditional likelihood ratio test, the
S test proposed by \cite{SW} is the only one that controls size well
when increasing the number of instruments, so we focus on this test
in our comparison.

We consider three versions of S, with 1, 3, or 7 instruments
in the polynomial and linear models. These instruments were obtained by
fitting piecewise linear functions on intervals defined by the
quartiles of $Z$, e.g., the three considered instruments are
$1(z\leq 0)$, $z \times 1(z\leq0)$, and $z$. For the group
heterogeneity model,  we implement S based on a reduced form with 3
instruments, namely the continuous variable $Z_1$, the discrete one
$Z_2$, and an interaction term. We then increase the number
of instruments to 7 and 15. We construct these instruments as
piecewise linear and interaction terms on intervals defined by the
quartiles of $Z_1$, e.g., the seven considered instruments are
$1(z_1\leq0)$, $z_1 \times 1(z_1 \leq 0)$, $z_1$, $z_2 \times 1(z_1
\leq 0)$, $z_2$, $z_2\times z_1 \times 1(z_1\leq 0)$, $z_2\times
z_1$. The critical value for the S statistic is  the quantile
from the  chi-square distribution with degrees of freedom equal to the
number of instruments.
For ICM, we compute p-values based on
499 simulated values of the statistic for each
considered hypothesis $H_0:\beta=\bar\beta$. For our new
test based on HICM, we compute 1999 simulated values of the
statistic, which are independent of the value of $\beta$ under $H_0$.

Our benchmark is the heteroskedastic version of the polynomial model
with a degree of weakness $c = 3$ and a sample size $n = 201$. We then
increase $c$ to 7, then we consider a linear model with $c=3$, and a
polynomial group model with $c=3$ and $n=401$.  The empirical sizes
are reported in Table \ref{table:sizes}, for one variable $Z$ with
$n=201$, or two variables $Z_1$ and $Z_2$ with $n=401$ (the
empirical levels under the null hypothesis $H_0: \beta = 0$ do not
depend upon other simulation details).  Size is controlled by the
three procedures, but HICM tends to be more conservative than
S, and ICM is the most conservative.

\begin{table}[!ht]
\centering
\scalebox{.9}{
\begin{tabular}[ht]{lrrrrr}
\hline
One variable $Z$ and $n=201$ & HICM & ICM & S-1IV & S-3IV & S-7IV\\
\hline
 & 4.12 & 3.10 & 4.98 & 4.72 & 4.14\\
 & 7.18 & 6.02 & 9.62 & 9.82 & 9.44\\
 \hline
Two variables $Z_1$ and $Z_2$ and $n=401$  & HICM & ICM & S-3IV & S-7IV & S-15IV\\
  \hline
 & 4.50 & 3.72 & 5.12 & 4.48 & 3.62\\
 & 7.98 & 8.02 & 9.92 & 9.22 & 8.74\\
\hline
\end{tabular}
}
\caption{\small Empirical sizes for HICM, ICM, and S, at  nominal  level of 5\% or 10\%.}
\label{table:sizes}
\end{table}

In Figure \ref{Fig i}, we represent the power curves of the different tests
when testing the null hypothesis that $H_0:\beta=\bar\beta$ with
$\bar\beta\in[-1.5,1.5] \times \sqrt{n} /c$.  In the benchmark
polynomial model, S with 1 IV does not have more than trivial power.
The power of S with 3 instruments is comparable to
the one of HICM. With 7 instruments, it decreases and is comparable to
the power of ICM.
Increasing identification strength does not improve the power of S
with only one IV.
With a linear specification, S with one IV is the
most powerful procedure as expected. Interestingly, HICM and ICM have
power comparable to S with 3 IV, and dominate S with 7 IV.  In  a
model with group heterogeneity, S with 3 IV does not have any power, while S
with 7 IV is slightly dominated by HICM.  Increasing to 15 IV lowers
the power of S, which becomes comparable to that of ICM.  Overall, HICM
performs extremely well and is at the same time  less
undersized and more powerful than ICM.  By contrast, S
 lacks power when too few instruments are used, while its power
worsens if too many instruments are used. As it is never clear how
to select the (number of) instruments, our HICM inference procedure appears
particularly attractive.

\begin{figure}[!ht]
\begin{center}
\includegraphics[width=9cm,height=6cm]{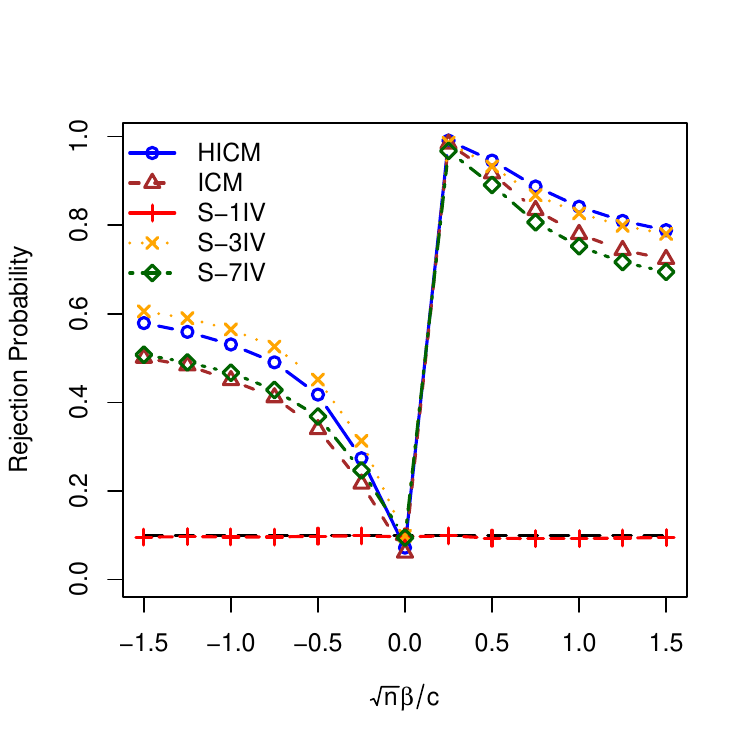}
\includegraphics[width=9cm,height=6cm]{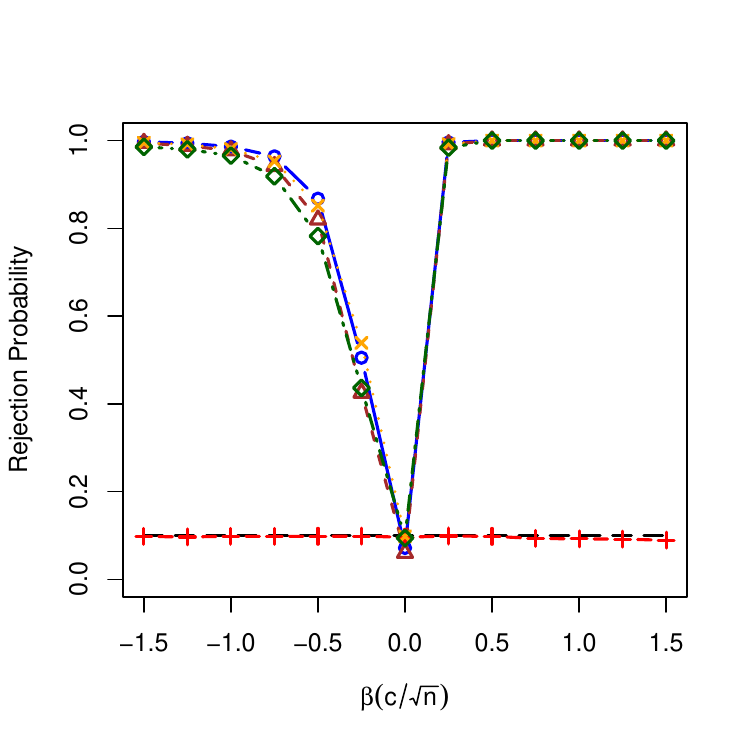}
\includegraphics[width=9cm,height=6cm]{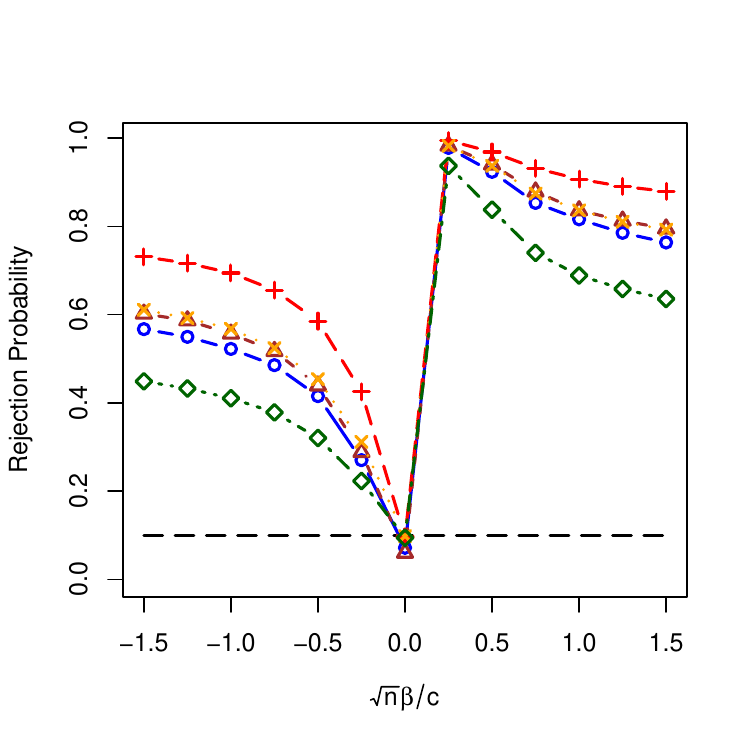}
\includegraphics[width=9cm,height=6cm]{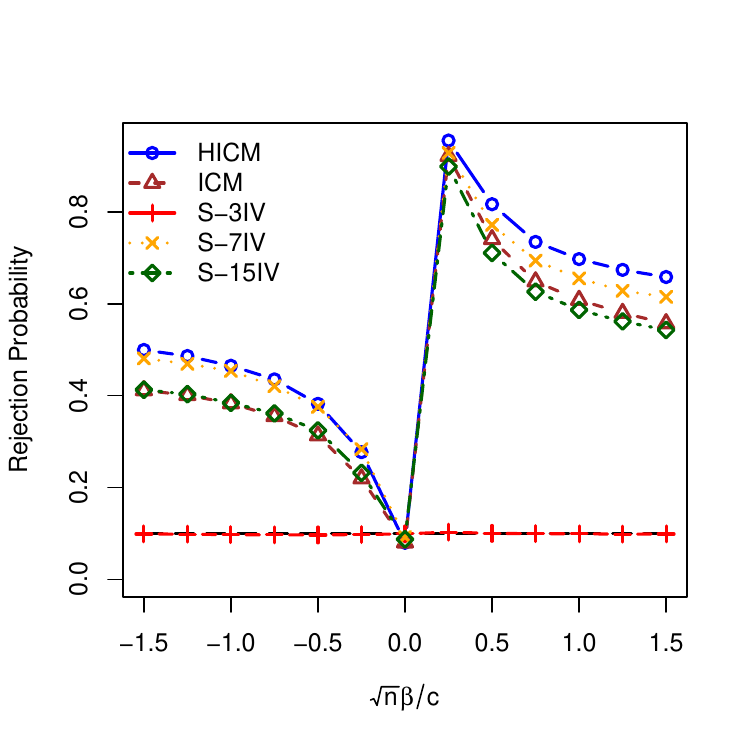}
\end{center}
\caption{\small Power curves for polynomial model with $c=3$ and $n=201$ (top left), with
stronger identification $c=7$ (top right),  linear model with $c=3$
and $n=201$ (bottom left),
polynomial group model with $c=3$ and $n=401$ (bottom right).}
\label{Fig i}
\end{figure}

\subsection{Specification testing} \label{subsec:sim specif test}

We investigate data-generating processes similar to the ones above,
but with $\delta$ varying away from zero.
There are not many competitors that can be compared to our specification test.
We consider two procedures relying
on unconditional moments obtained from a given set of instruments.
First, the J-test of overidentification based on the
continuously updated GMM with conservative critical values, hereafter
JCUE, as recently discussed in
\cite{AntoineFrazierRenault2026WP}\footnote{\cite{AntoineFrazierRenault2026WP}
also propose a conditional approach that exploits information about
the existence of strongly identified directions in the parameter space
to define data-dependent critical values that are less conservative
than the robust ones. We do not explicitly consider this approach here.}.
Regardless of identification strength, the CUE objective
function is asymptotically upper-bounded by a random variable
distributed as a chi-square with degrees of freedom equal to the
number of instruments, so one obtains simple conservative critical
values. We also compute the jackknife T-specification test proposed
by \cite{Chao2014}, hereafter JackT.
The  test is based on (i) estimating \(\beta_0\) by HFULL, a
heteroskedasticity-robust version of the Fuller estimator proposed by
\cite{hausman_instrumental_2012}, and (ii) considering a jackknife
version of the overidentification statistic, based on the objective
function of the JIVE2 estimator proposed by
\cite{AngristImbensKrueger1999}.  The associated  statistic is
asymptotically distributed as a chi-square with degrees of freedom
equal to the degree of overidentification, when the number of
instruments increases with the sample size, or when it is fixed under
maintained homoskedasticity. Hence, while the test is generally valid under the framework of
{\em many weak instruments}, it may not be in the  setup considered here.

\begin{table}[ht]
\centering
\scalebox{.8}{
\begin{tabular}[ht]{lrrrrrrr}
\hline
 & HICM$^*$ & JCUE-1IV & JCUE-3IV & JCUE-7IV &  JackT-3IV & JackT-7IV &
 \\
\hline
Polynomial & 1.26 & 0.04 & 1.48 & 1.44 &   8.40 & 4.84 & \\
 & 2.72 & 0.10 & 3.40 & 3.84 &  15.70 & 9.74 & \\
Strong.id. & 1.82 & 0.00 & 2.00 & 2.24 &   8.14 & 5.20 & \\
 & 3.40 & 0.04 & 4.22 & 5.06 & 14.94 & 9.94 & \\
Linear & 1.10 & 0.00 & 1.44 & 1.26 &  4.18 & 4.64 & \\
 & 2.22 & 0.00 & 3.62 & 3.76 &  7.84 & 8.28 & \\
\hline
 & HICM$^*$ & JCUE-3IV & JCUE-7IV & JCUE-15IV & JackT-3IV & JackT-7IV & JackT-15IV\\
\hline
Polynomial group & 0.86 & 0.02 & 1.56 & 1.16 & 1.72 & 4.84 & 4.42\\
 & 2.66 & 0.12 & 3.64 & 4.10 & 3.98 & 10.10 & 8.34 \\
 \hline
 \end{tabular}
}
\caption{\small Empirical sizes for  HICM$^*$,
 JCUE, and JackT, at  nominal  level of 5\% or 10\%.}
\label{table:sizesspec}
\end{table}

\begin{figure}[!ht]
\begin{center}
\includegraphics[width=9cm,height=6cm]{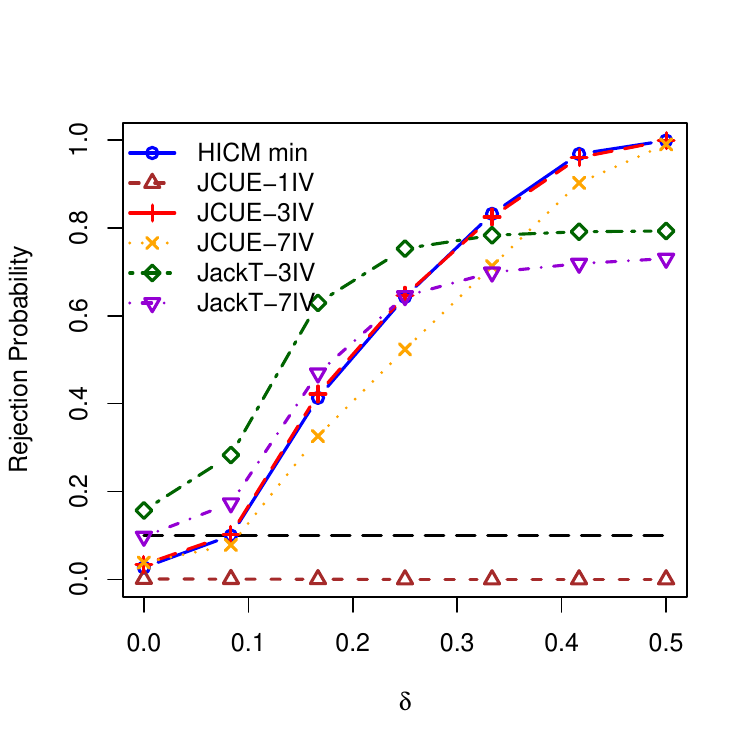}
\includegraphics[width=9cm,height=6cm]{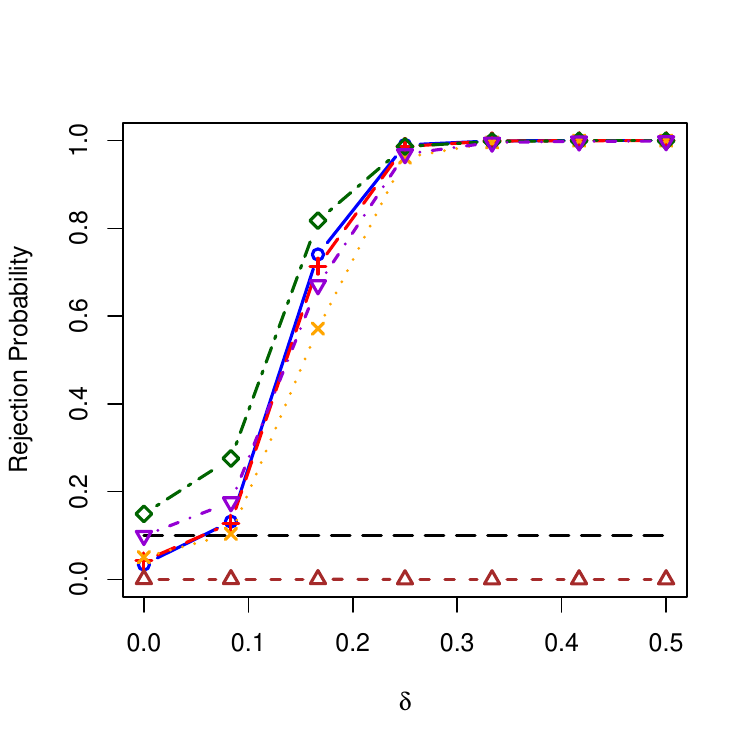}
\includegraphics[width=9cm,height=6cm]{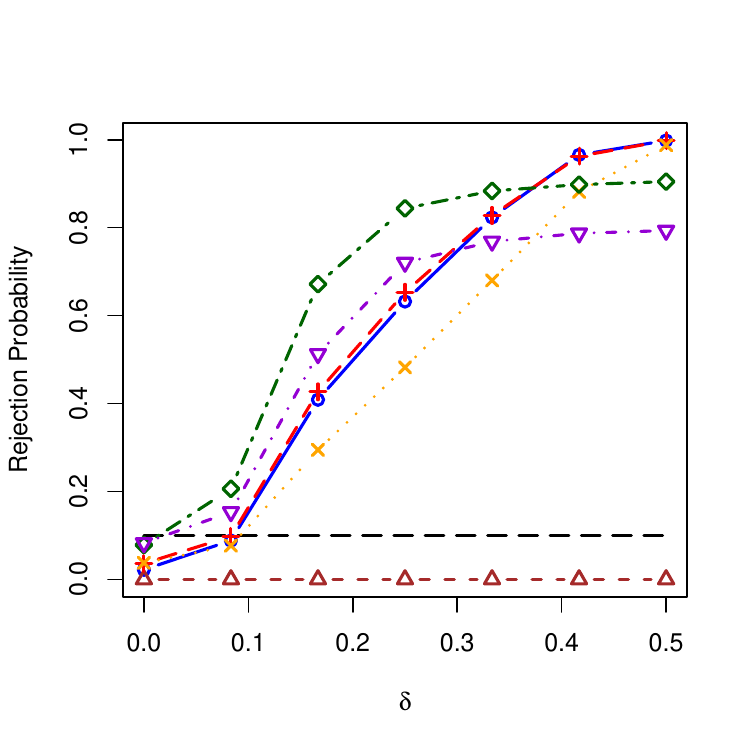}
\includegraphics[width=9cm,height=6cm]{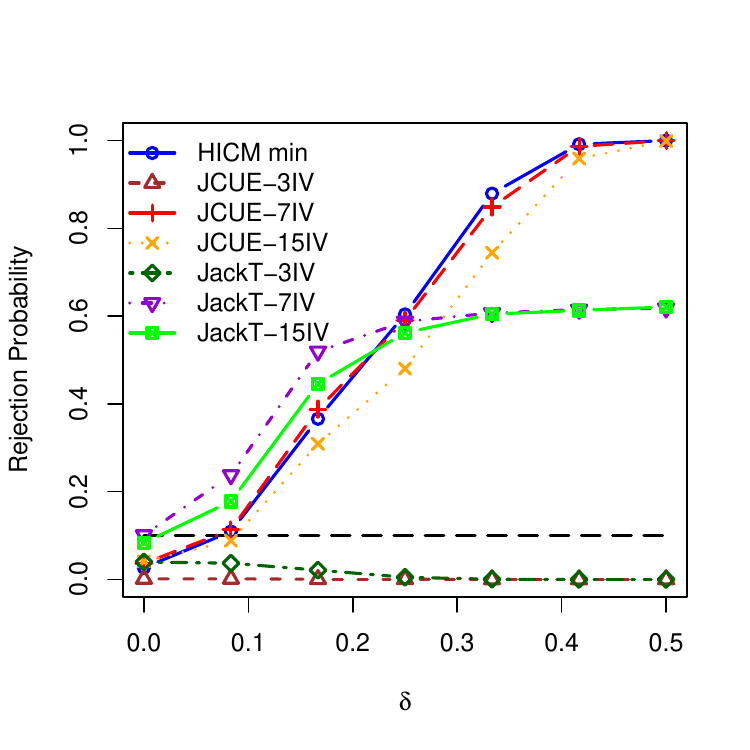}
\end{center}
\caption{\small Power curves for polynomial model with $c=3$ and $n=201$ (top left), with
stronger identification $c=7$ (top right),  linear model with $c=3$
and $n=201$ (bottom left),
polynomial group model with $c=3$ and $n=401$ (bottom right).}
\label{Fig ii}
\end{figure}

We consider different versions of JCUE and JackT, respectively with 1, 3, and 7 instruments
for one variable $Z$, and 3, 7, or 15 instruments for two
variables $Z_1$ and $Z_2$, see Section \ref{subsec:sim idr
inference} for details (note that JackT cannot be used with one
instrument only, as there is exact identification).
The empirical sizes are reported in Table
\ref{table:sizesspec}. Our specification test and the tests
based on JCUE are undersized as expected.  The test based on JackT
with 3 instruments is oversized for the polynomial model, but size
controls improve with 7 instruments, as well as for the linear and
polynomial group models.  In Figure \ref{Fig ii}, we present power
curves. For our benchmark polynomial model, the JCUE test has trivial
power with one instrument. Power improves with 3 instruments, then
decreases with 7 instruments. Our HICM$^*$ test has similar power to
the best of the JCUE tests.  The JackT tests are more powerful than
our test for small values of $\delta$, but their power curves flatten
for higher values. As for JCUE, increasing the number of instruments
negatively affects power. In the stronger identification setup, all
tests are powerful, but JCUE with only one instrument.  In the linear
model and the polynomial group model, we found the same main features
as in the benchmark model. However, JackT with 3 instruments has
trivial power for the polynomial group model.  Throughout the
different designs, the small sample properties of both JackT and JCUE
are affected by the number of instruments. By constrast, our
specification test consistently controls size and is powerful
regardless of identification strength and the particular functional
form that links instruments and endogenous variables.

\section{Empirical illustration}
\label{sec.ea}

We revisit \cite{antoine_identification-robust_2023}'s empirical
application which extends some of the results presented in
\cite{Mexico}. They study the impact of a large population collapse in
16th-century Mexico on land institutions by adopting an
instrumental-variables empirical strategy based on the characteristics
of a massive epidemic in the mid-1570s, which is believed to have been
caused by a rodent-transmitted pathogen that emerged after several
years of drought were followed by a period of above-average rainfall.
As explained by these authors, the sharp decline in population lowered
the costs and increased the benefits of acquiring land from indigenous
villages in many areas. As their three excluded instrumental
variables, they use proxies for climate conditions : (i) drought, the
sum of the 2 lowest consecutive PDSI values between 1570 and 1575
(more negative numbers indicate severe and prolonged drought), (ii)
rainfall, the maximum PDSI between 1576 and 1580 (as a measure of
excess rainfall), and (iii) gap, the difference between the minimum
PDSI between 1570 and 1575 and the maximum between 1576 and
1580.\footnote{The Palmer Drought Severity Index (PDSI) is a
normalized measure of soil moisture that captures deviations from
typical conditions at a given location.}

Using the data constructed by \cite{Mexico}, we estimate the
short-term effects of the population collapse in the model
\begin{equation} \label{eq:mexicomainlinear}
S1:\quad y_i = \beta_0 + \beta_1 Y_{2i} + \gamma' X_{1i} + u_i \, ,
\qquad E(u_i|X_{1i},X_{2i})=0 \, ,
\end{equation}
where $y_i$ is the inverse hyperbolic sine of the percent rural
population living in hacienda communities in 1900, $Y_{2i}$ is the
population decline in municipality measured as the log ratio of 1650
and 1570 density, $X_{2i}$ represents the vector of the climate
instruments, and $X_{1i}$ is a vector of control variables of
geographic and economic features, namely the standard deviation of
PDSI, a measure of maize productivity, various measures of elevation
and slope, as well as the log of tributary density in 1570 and
governorship-level fixed effects, see Column 6 in Table 2 in
\cite{Mexico}.

In Table \ref{table:Mexico}, we report confidence regions obtained
with three inference procedures that simultaneously test the null and
the validity of the model: (i) HICM introduced in this paper, (ii) ICM
introduced in \cite{antoine_identification-robust_2023}, and (iii) S introduced in
\cite{SW} assuming a linear reduced form. Following the analysis in
\cite{antoine_identification-robust_2023}, we use as instruments the
two most reliable variables gap and drought, as well as regional
dummies as additional controls.\footnote{When using the three climate instruments with
regional dummies, the model is rejected.}  All three procedures yield
a significant and negative effect of the population decline on the
percent rural population living in hacienda as expected: a decrease in
the ratio of 1650 to 1570 density increases the likelihood of having
more large estates per area in 1900. The confidence set obtained with
HICM is the narrowest and  is fully enclosed in the ones obtained
with S or ICM. Computationally, the advantage of HICM over ICM is clear. Since its
critical value does not depend on the parameter values, it only needs
to be computed once for the entire grid of candidates. Accordingly, and
similarly to S, the associated confidence set can be obtained much
faster than by ICM.  With two instruments and  a grid of 2,500
candidate points,  it takes under 50 seconds for HICM, but over 52 minutes for ICM! 

\begin{table}
\begin{center}
\scalebox{.9}{
\begin{tabular}{rcr}
\hline
&  \textbf{2 climate IV} &  \textbf{Computational time in sec.} \\
\hline
HICM &  [-1.235, -0.797] & 46.66 \\
ICM  &  [-2.159, -0.416] & 3,140.22\\
S &     [-1.369, -0.682] & 46.66 \\ \hline
\end{tabular}
}
\caption{\small 95\% confidence intervals for the population collapse
 obtained using a grid of 2,500 evenly spread points over $[-2.5,0]$ and 499 bootstraps to generate critical values.}
\label{table:Mexico}
\end{center}
\end{table}

While \cite{Mexico} and \cite{antoine_identification-robust_2023}
assume throughout that the effect of the population collapse $Y_{2}$
on land tenure $y$ is linear, we relax this assumption by considering
 instead  nonlinear specifications with either 2 or 3 endogeneous
 variables, specifically
\begin{eqnarray*}
S2: \quad  y_i &= & b_0 + b_1 Y_{2i} \times I_{Y_{2i}<-1.89} + b_2 Y_{2i} \times I_{Y_{2i}>-1.89}
+ \gamma_b' X_{1i} + u_{b,i}  \label{eq:mexico nonlinear 2endo}  \\
S3: \quad  y_i &= & c_0 + c_1 Y_{2i} \times I_{Y_{2i}<-1.89} + c_2
Y_{2i} \times I_{-1.89<Y_{2i}<-0.89} + c_3 Y_{2i} \times I_{Y_{2i}>-0.89}\\
&& \qquad
+ \gamma_c' X_{1i} + u_{c,i}  \label{eq:mexico nonlinear 3endo}
\end{eqnarray*}
Both specifications explore whether more negative values of the
population collapse - that is, more severe population declines - are
driving the previous results.\footnote{We thank I. Andrews for this
suggestion.}  Specification S2 introduces an indicator which targets observations below/above the 20\%
quantile of $Y_2$ (equal to -1.89), whereas Specification S3 includes
an additional indicator to target observations above the 80\% quantile
of $Y_2$ (equal to -0.89).

Multivariate confidence regions are obtained using a two-dimensional
grid of evenly spread points over $[-1.5,-0.5]\times[-8,1]$ (with
respective steps of 0.01 and 0.05) for Specification S2 and a
three-dimensional grid of evenly spread points over $[-2.7,
0.2]\times[-7.5,0]\times[-18, 4.5]$ (with
respective steps 0.05, 0.05 and 0.1) for Specification S3. Other details are similar to
above. We only consider the procedures based on HICM and S: based on
our previous results, it would take weeks to obtain a multivariate
confidence region with ICM for Specification S3. For both
specifications, HICM provides finite and bounded
confidence regions, while S provides large and possibly unbounded ones.
In Figure \ref{graph: mexico NL specif}, we  report the 95\% confidence
regions for the parameters of  Specification S2. The HICM confidence region is
small and bounded, while that obtained with S is much larger and
possibly unbounded.\footnote{As a robustness check,  we consider much larger grids of
candidate points.  Figure \ref{graph: mexico NL S extended} in our
Appendix \ref{sec:appfig} strongly suggests that the S confidence region is indeed
unbounded.} In Table \ref{table:Mexico NL}, we report the 95\%
confidence intervals obtained with HICM by projecting the multivariate
confidence regions.  With Specification S2, the effect of the
population collapse below the quantile at 20\% is negative,
significant, and in line with our results for Specification S1. The
effect above the 20\% quantile is also negative and significant,
potentially larger in magnitude, but also less precisely estimated.
Given the width of confidence intervals, the previously estimated
linear Specification S1 cannot be ruled out.
With  Specification S3, the confidence intervals  are in line with
previous specifications, but only $c_2$ is significantly different from 0 at
95\% confidence level.


\begin{figure}[!ht]
\begin{center}
\includegraphics[width = 6cm, height = 4cm]{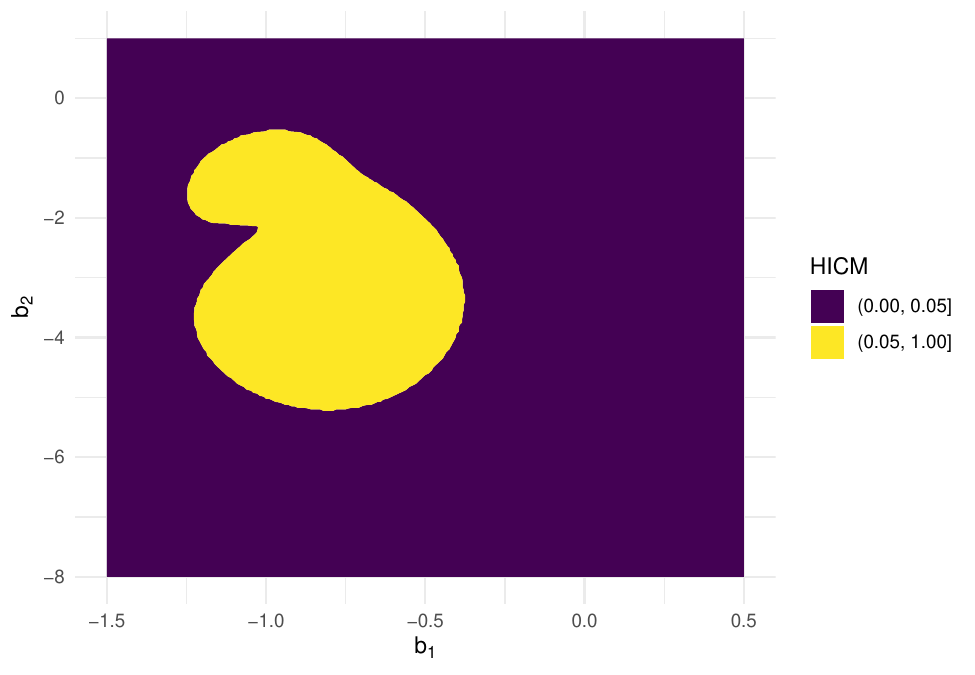}
\includegraphics[width = 6cm, height = 4cm]{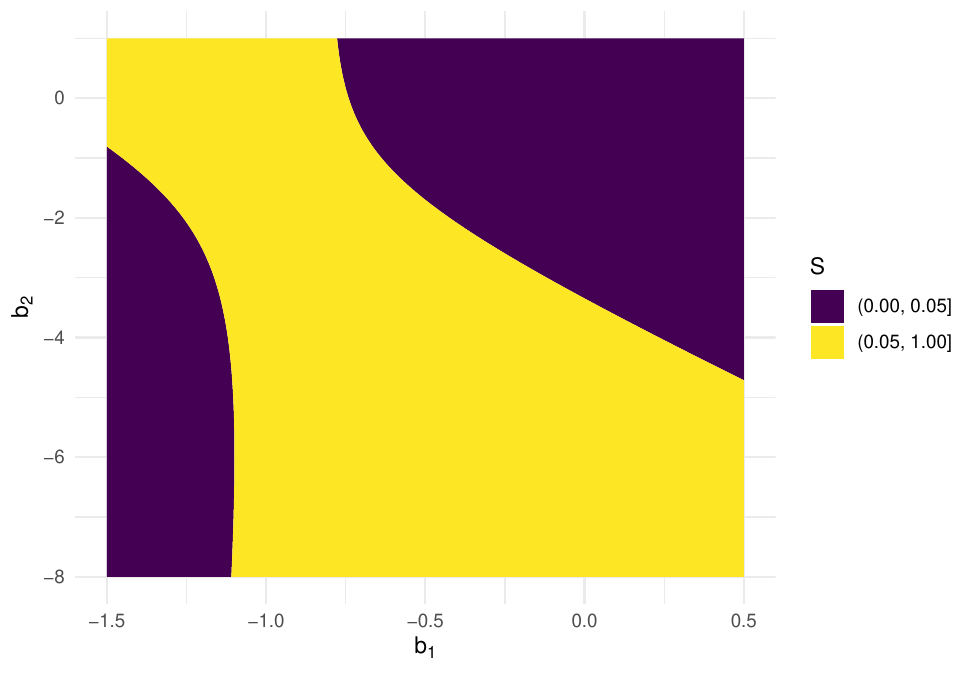}
\end{center}
\caption{\small 95\% confidence regions for the (nonlinear) effect of
the population collapse obtained with HICM (left) and S (right).}
\label{graph: mexico NL specif}
\end{figure}

\begin{table}[ht]
\begin{center}
\begin{tabular}{c|r|c|r|c|r|c|}
\hline
Specification & \multicolumn{2}{c|}{Linear } & \multicolumn{2}{c|}{S2 } & \multicolumn{2}{c|}{S3} \\ \hline
&$\beta_1$ & [-1.23, -0.80]
&$b_1$ &  [-1.24, -0.38] & $c_1$ & [-2.50, 0.15] \\
&&&$b_2$ &  [-5.20, -0.55] & $c_2$ & [-7.20, -0.05] \\
&&& & & $c_3$ & [-17.40, 4.40] \\
\hline
\end{tabular}
\caption{\small 95\% confidence intervals for the (nonlinear) effect of the
population collapse with HICM.}
\label{table:Mexico NL}
\end{center}
\end{table}

To put these results into perspective, we revisit part of the
counterfactual analysis in \cite{Mexico}. Specifically, to obtain what
landholdings would have been in the absence of a population collapse,
we subtract off the predicted marginal effect of the population change
in each municipality from the actual 1900 outcome.
We find that the distribution of hacienda population changes
substantially under our counterfactual. Our results are reported in
Table \ref{table:Mexico counterfactual} for the median and the
third quartile of the percentage of 1900 population living in haciendas,
respectively 16.7\% and 44.4\%. When we remove the effect of the
population collapse given by our estimation strategy, we find that the
percentage of 1900 population living in haciendas drops by 1.5\% to
4.6\% with  Specification S1,  and by 1.1\% to 9.1\% with
Specification S2. The drop is  potentially larger at the third
quartile with a lower bound of 2.7\% for the linear specification and
2.4\% for Specification S2. Overall, such an impact is economically
meaningful and practically relevant.

\begin{table}[ht]
\begin{center}
\begin{tabular}{lcc}
\hline
& \multicolumn{2}{c}{Counterfactual} \\ \hline
\% of 1900 population living in haciendas & Linear & S2  \\ \hline
Median       & [1.5, 4.6]  & [1.1,9.2]\\ \hline
3-rd quartile & [2.7, 13.6] & [2.4,24.7] \\ \hline
\end{tabular}
\caption{\small Counterfactual analysis of the causal impact of the
demographic collapse: we report the predicted marginal effect of the
population change in each municipality  at the median and at the third
quartile of the percentage of 1900 population living in haciendas.}
\label{table:Mexico counterfactual}
\end{center}
\end{table}

\section{Concluding remarks} \label{sec.concl}

We have developed and studied a new set of inference procedures in a
linear IV model based on HICM, an heteroskedastic version of the ICM
statistic of \cite{antoine_identification-robust_2023}.  The tests are
easy to implement, robust to any identification pattern and unknown
heteroskedasticity, and nonparametric with respect to the first-stage
equation. The key innovative feature is the independence of critical
values from parameter values.  As a result, inference is greatly
facilitated. In addition, this also allows for straightforward
subvector inference and pure specification testing.  In simulations,
our procedure is competitive with the ICM test of
\cite{antoine_identification-robust_2023}.  While subvector inference
is conservative, we have shown that in an application that it can
yield bounded and informative confidence regions for parameters of
interest.

\medskip
\bibliography{biblio}

\newpage
\appendix

\phantomsection\label{supplementary-material}

\begin{center}
{\large\bf SUPPLEMENTARY MATERIAL}
\end{center}

\section{Proofs}
\label{sec:proofs}

\subsection{Proof of Theorem \ref{th:level}}

The proof is done in three main steps. First, we show uniform convergence
of processes of the type $n^{-1/2} \sum_{i=1}^{n}{\left(b'
\widehat{\Omega}(Z_i) b \right) \left( Y_{i} - \E(Y_{i}| Z_{i}
)\right) \exp(i s'Z_{i})}$. We then introduce some notations and
preliminary results.  The main proof  shows uniform size control by
contradiction.

\subsubsection{Uniform convergence of processes}
\label{sec:ucp}

From Assumption \ref{ass:iid}, $Z$ is bounded, and from Assumption \ref{ass:w}, $\mu$
has a compact support $M$.
The class of functions $\left\{ s'Z, s \in M \right\}$ has
Vapnik-\v{C}ervonenkis dimension $k+2$ and thus has bounded uniform
entropy integral (BUEI) with constant envelope.
Since the functions $t \rightarrow \cos(t)$ and $t \rightarrow
\sin(t)$ are bounded Lipschitz with
derivatives bounded by 1, the class $\left\{ \cos(s'Z), \sin(s'Z), s
\in M \right\}$ is BUEI with constant envelope, see \citet[Lemma 9.13]{kosorok}.

By Assumption \ref{ass:pi}, the class ${\cal E}$ is BUEI with envelope
$F$. From \citet[Theorem 9.15]{kosorok}, the class $\left\{ \Pi (Z) \cos(s'Z),
\Pi (Z) \sin(s'Z),  \Pi (\cdot)  \in {\cal E},  s \in M \right\}$ is
BUEI with envelope $c F$ for some constant $c$, and from \citet[Lemma 2.8.3]{vaart_weak_2000}
\begin{align*}
\left(
\begin{array}{c}
  n^{-1/2} \sum_{i=1}^{n}{  \left[ \E(Y_{i}| Z_{i} )  \cos(s'Z_{i})
  - \E \left( Y \cos(s'Z) \right) \right] }
 \\
  n^{-1/2} \sum_{i=1}^{n}{ \left[ \E(Y_{i}| Z_{i} ) \sin(s'Z_{i})
    - \E \left( Y \sin(s'Z) \right) \right]}
\end{array}
\right)
 & \leadsto
 \left(
\begin{array}{c}
 \mathbb{G}_{1} (s)
 \\
 \mathbb{G}_{2} (s)
 \end{array}
\right)
\, ,
\end{align*}
uniformly in $P \in {\cal P}$ where $\left( \mathbb{G}'_{1}(\cdot),
\mathbb{G}'_{2}(\cdot)\right)$  is a vector Gaussian process with mean $\mO$.
Formally weak convergence uniform in $P$ means that
\begin{align*}
\sup_{P \in {\cal P}} d_{BL}(\mathbb{G}_{n},\mathbb{G}) & \rightarrow 0
\quad \mbox{where } \ d_{BL}(\mathbb{G}_{n},\mathbb{G})  = \sup_{f \in BL_{1}}
\left| \E f\left(\mathbb{G}_{n}\right) - \E f \left( \mathbb{G}
\right) \right|
\end{align*}
is the bounded Lipschitz metric, that is $BL_{1}$ is the set of real
functions bounded by 1 and whose Lipschitz constant is bounded by 1.
This implies that
\begin{equation*}
  n^{-1/2} \sum_{i=1}^{n}{  \left[ \E(Y_{i}| Z_{i} )  \exp(i s'Z_{i})
  - \E\left(Y \exp(is'Z)\right) \right]}
\leadsto  \mathbb{G}_{1} (s) +  i \ \mathbb{G}_{2} (s)
\label{lem:eq1}
\end{equation*}
As $\E \|Y\|^{2+\delta} < \infty$,
\begin{equation}
  n^{-1/2} \sum_{i=1}^{n}{ \left( Y_{i} - \E(Y_{i}| Z_{i} )\right)
 \exp(i s'Z_{i})}
\leadsto  \mathbb{G}_{3} (s) +  i \  \mathbb{G}_{4} (s)
\, ,
\label{lem:eq2}
\end{equation}
as the above class of functions is BUEI for a squared-integrable envelope.

Since $\Omega(\cdot)$ is a variance matrix with uniformly bounded
elements, the functions $b'\Omega(\cdot) b$ for $\|b\| \leq M$,
and $\Omega \in {\cal O}$ satisfies
\begin{align*}
\left| b' \Omega_{1}(\cdot) b  - b' \Omega_{2}(\cdot) b \right| &
\leq M^{2} \|\Omega_{1} - \Omega_{2}\|
\, .
\end{align*}
From Assumption \ref{ass:variance} and \citet[Lemma 9.13]{kosorok},
these functions forms a BUEI class with constant envelope. Consider ${\cal B} =
\left\{ \left( b'\Omega(\cdot) b\right)^{-1/2}, \|b\| \leq M, \Omega \in {\cal O}
\right\}$. Since the function $\phi (f) = f^{-1/2} $ is Lipschitz for
$f$ uniformly bounded away from zero, ${\cal B}$ is a BUEI class with
constant envelope.
Gathering  results, for $B \in {\cal B}$,
\begin{equation*}
\mathbb{G}_{n} (B,s) =
n^{-1/2} \sum_{i=1}^{n}{B (Z_{i}) \left( Y_{i} - \E(Y_{i}| Z_{i} )\right)
\exp(i s'Z_{i})}
\leadsto \mathbb{G} (B,s)
\, ,
\label{lem:eq3}
\end{equation*}
uniformly in $P \in {\cal P}_{\beta_0}$ and $\beta_0 \in {\cal B}$,
where $\mathbb{G} (B,s)$ is a centered Gaussian vector process.

Replacing
$B(\cdot) = \left( b' \Omega(\cdot) b \right)^{-1/2}$ by $\widehat{B}(\cdot) =
\left( b' \widehat{\Omega}(\cdot) b\right)^{-1/2}$, does
not change the uniform weak limit of the process.  Indeed, from Assumption
\ref{ass:variance}-(iii) and (iv), it is sufficient to show that
\begin{equation}
\sup_{P \in {\cal P}}
\Pr \left[  \sup_{m\geq n} \sup_{s} \| \mathbb{G}_{m} (\widehat{B}_{m},s) -
\mathbb{G}_{m} (B,s) \|_{\cal B}  > \varepsilon \right]
\rightarrow 0
\qquad \forall \varepsilon >0
\, .
\label{eq:asyequicont}
\end{equation}
This follows as $\mathbb{G}_{n}(B,s) $ is  asymptotically
equicontinuous uniformly in $P$, see \citet[Theorem
2.8.2]{vaart_weak_2000}.
\medskip

\subsubsection{Notations and preliminary results}

For  vector complex-valued functions $h_{1}(s)$ and $h_{2} (s)$,
define the scalar product
\[
\ip{h_{1}}{h_{2}}  = \frac{1}{2}
\left(
  \int_{}^{}{  \left( \overline{h}'_{1}(s) {h}_{2}(s) + h_{1}'(s)
  \overline{h}_{2} (s) \right)
  \, d\mu (s)}
\right)
\]
and the norm
$ \| h_{1} \|  = \ip{h_{1}}{h_{1}}^{1/2}$.
Denote
\[
\widehat{h}_{\beta_{0}} (s) \equiv n^{-1/2} \sum_{i=1}^{n}{
\widehat{S}_{i} \exp(i s'Z_{i})}
\, ,
\quad \widehat{S}_{i} = \left( b_0'\widehat{\Omega}(Z_i) b_0 \right)^{-1/2}
Y_i'b_0
\, .
\]
From  (\ref{lem:eq2}), $\widehat{h}_{\beta_0}(s) \leadsto  \mathbb{G}_{S}(s)$
uniformly in $P \in {\cal P}_{\beta_0}$ and in $\beta_{0} \in {\cal B}$, where
$\mathbb{G}_{S}(s)$ is a centered complex Gaussian process with
covariance $
E \mathbb{G}_{S}(s) \mathbb{G}_{S}(t) = E  \exp(i s'Z)  \exp(i t'Z)$.
Our test statistic writes $\HICM (\beta_{0}) = \HICM (\widehat{h}_{\beta_{0}}) =
\|\widehat{h}_{\beta_0}\|^{2} = \widehat{S}'W\widehat{S}$, where
$\widehat{S} = \left( \widehat{S}_{1}, \ldots, \widehat{S}_{n} \right)'$.

\begin{lem}
Over the set $\left\{ h : \|h\| \leq C\right\}$,
 $\HICM (h)$ is bounded and Lipschitz continuous in $h$.
\label{lem:lip}
\end{lem}
\begin{proof}
Boundedness is trivial. For Lipschitz continuity,
\begin{multline*}
\left| \HICM (h_{1}) - \HICM (h_{2}) \right|
= \left| \|h_{1}\|^{2} - \|h_{2}\|^{2} \right|
 = \left| \ip{h_{1}-h_{2}}{h_{1} + h_{2}} \right|
\\
 \leq \|h_{1} - h_{2} \| \|h_{1} + h_{2}\|
 \leq \|h_{1} - h_{2} \| ( \|h_{1}\| + \|h_{2}\| ) \leq 2\, {C} \,
\|h_{1} - h_{2} \|
\, .
\tag*{\qedhere}
\end{multline*}
\end{proof}

\begin{lem} Under Assumption \ref{ass:iid} and \ref{ass:w},
\[
\lim_{M \rightarrow \infty}\sup_{\beta_{0}} \sup_{P \in {\cal P}_{\beta_{0}}}
\Pr \left[ \HICM (\beta_{0}) > M \right] \rightarrow 0
\, .
\]
\label{lem:boundicm}
\end{lem}
\begin{proof}
Let $S_i = \left( b_0' {\Omega}(Z_i) b_0 \right)^{-1/2} Y_i'b_0$ and
$h_{\beta_0}$ the corresponding empirical process. Then
\begin{align*}
\HICM (h_{\beta_{0}})
&=  S'WS = n^{-1} \sum_{i=1}^{n}{S_{i}^{2} w(0)}
+ n^{-1} \sum_{i=1}^{n}{ \sum_{j\neq i}^{}{}S_{i}S_{j} w(Z_{i}-Z_{j})}
\, .
\end{align*}
Hence, for some constants $C, C' , C''>0$ independent of $P \in {\cal
P}_{\beta_{0}}$ and of $\beta_{0} \in {\cal B}$,
\begin{align*}
\Pr \left[ n^{-1} \sum_{i=1}^{n}{S_{i}^{2} w(0)} > M/2 \right] & \leq
2  w(0) \frac{\E S_{1}^{2}}{M} \leq \frac{C}{M}
\\
\Pr \left[ n^{-1} \sum_{i=1}^{n}{ \sum_{j\neq i}^{}{}S_{i}S_{j}
w(Z_{i}-Z_{j})}  > M/2 \right] & \leq
4 C' \frac{\E^{2} (S_{1}^{2}) }{M^{2}} \leq \frac{C''}{M}
\, ,
\end{align*}
using the boundedness of  $w(\cdot)$ and  Markov's inequality.
Now
\[
\HICM (\widehat{h}_{\beta_0}) -  \HICM (h_{\beta_{0}})
= \| \widehat{h}_{\beta_0} - h_{\beta_{0}}\|^2  + 2
\ip{\widehat{h}_{\beta_0} - h_{\beta_{0}}}{h_{\beta_{0}}}
\, ,
\]
where
\[
\widehat{h}_{\beta_0} - h_{\beta_{0}}  =
n^{-1/2} \sum_{i=1}^{n}{
\left[
\left(b_0' \widehat{\Omega}(Z_i) b_0\right)^{-1/2}
- \left(b_0' {\Omega}(Z_i) b_0\right)^{-1/2}
\right]
Y_{i}'b_0  \exp(i s'Z_{i})}
\, .
\]
Applying a analogous reasoning as in Section \ref{sec:ucp} shows that
the above quantity
is an $o_p(1)$ uniformly in $P \in {\cal P}_{\beta_{0}}$ and in $\beta_{0}$.
Hence
$
\HICM (\widehat{h}_{\beta_0}) -  \HICM (h_{\beta_{0}, {S}}) = o_p(1)
$
uniformly in $P \in {\cal P}_{\beta_{0}}$ and in $\beta_{0}$.
\end{proof}

\subsubsection{Main proof}

Let $G  \sim N (0, \mI)$. From our results in Section \ref{sec:ucp},
$
h_{{G}} (s) = n^{-1/2} \sum_{i=1}^{n}{ {G}_{i}  \exp(i
s'Z_{i})}  \leadsto \mathbb{G}_{S}(s)
$
uniformly in $P \in {\cal P}$. We say that $\widehat{h}_{\beta_0}$ {\em
uniformly weakly converges} to $h_{{G}}$  in $P \in
{\cal P}$, i.e.
\[
\sup_{\beta_{0} \in {\cal B}}
\sup_{P \in {\cal P}_{\beta_0}} d_{BL}(\widehat{h}_{\beta_{0}}, h_{{G}})
\rightarrow 0
\, ,
\]
see \cite{kasy_uniformity_2019} for a similar  terminology.

Let  $F (x) = \ind \left[ x < C_{1}\right] +
\frac{C_2-x}{C_{2}-C_{1}} \ind\left[ C_{1} \leq x \leq C_{2}\right]$
for some $0<C_{1} < C_{2}$ and consider the continuous truncation of
$\HICM(h)$ defined by
$
\HICM_{F} ( h) = \HICM (h) F(\|h\|)
$.
Consider the quantile of $\HICM_{F} (h)$
\[
c_{F,1-\alpha}(h) = \inf \left\{ c : \Pr\left[
\HICM_{F} (h) \leq c \right] \geq 1-\alpha \right\}
\, .
\]
Lemma \ref{lem:lip} ensures that $\HICM_{F}(h)$ is Lipschitz
continuous, so
\begin{align*}
1 - \alpha & \leq \Pr \left[ \ICM_{F} (h_1) \leq c_{F,1-\alpha}(h_1)
\right]
\leq Pr \left[ \ICM_{F} (h_2) \leq c_{F,1-\alpha}(h_1) + K\|h_{1}-h_{2}\|
\right]
\, ,
\end{align*}
for some constant $K>0$,
so that $c_{F,1-\alpha}(h_2) \leq c_{F,1-\alpha}(h_1) +
K\|h_{1}-h_{2}\|$.
Inverting the role of $h_{1}$ and $h_{2}$,
$c_{F,1-\alpha}(h_1) \leq c_{F,1-\alpha}(h_2) + K\|h_{1}-h_{2}\|$, so
 $c_{F,1-\alpha}(h)$ is Lipschitz in $h$.

Assume now that the conclusion of Theorem \ref{th:level} does not
hold. Then there exists some $\delta>0$, an infinitely increasing
subsequence of sample sizes $n_{j}$ and a sequence of probability
measures $P_{n_{j}} \in {\cal P}_{\beta_{0,n_{j}}}$, with corresponding
sequences of $\beta_{0,n_{j}} \in {\cal B}$ and $\Pi_{n_{j}}(\cdot)
\in {\cal E}$, such that
\[
\Pr_{n_{j}} \left[ \HICM (\widehat{h}_{\beta_{0,n_{j}}}) > c_{1-\alpha} (h_{{G}})
\right]
> \alpha + 3 \delta \qquad \forall n_{j}
\, .
\]
Using Lemma \ref{lem:boundicm}, choose $C_{1}$ such that
$
\Pr_{n_{j}} \left[  \HICM  (\widehat{h}_{\beta_{0,n_{j}}})
 \geq C_{1}  \right] < {\delta}
$. Now
\[
\Pr \left[ \HICM (\widehat{h}_{\beta_{0}}) > x \right] \leq \Pr\left[
\HICM_{F}  (\widehat{h}_{\beta_{0}}) > x \right]
+ \Pr \left[ \HICM (\widehat{h}_{\beta_{0}})
 \geq C_{1} \right]
\]
for any $x>0$,  $\beta_{0}$ and $P_{\beta_{0}}$.
As $c_{F,1-\alpha}(h) \leq c_{1-\alpha}(h)$,
\[
\Pr_{n_{j}} \left[ \HICM_{F}  (\widehat{h}_{\beta_{0,n_{j}}})
 > c_{F,1-\alpha} (h_{{G}})
\right]
> \alpha + 2 {\delta}
\qquad \forall n_{j}
\, .
\]
As $\HICM_{F} (h)$ is bounded and  Lipschitz in $h$, by the  uniform convergence
of $\widehat{h}_{\beta_{0}}$ to $h_{{G}}$,
\[
\sup_{\beta_{0}} \sup_{P \in {\cal P}_{\beta_0}} \sup_{x}
\left|
\Pr \left[ \HICM_{F} (\widehat{h}_{\beta_0}) > x \right]
-
\Pr \left[ \HICM_{F} (h_{{G}}) > x \right]
\right|
\rightarrow 0
\, .
\]
Therefore, for $n_{j}$ large enough,
$
\Pr_{n_{j}} \left[ \HICM_{F} (h_{{G}}) > c_{F,1-\alpha} (h_{{G}})
\right] \geq \alpha + {\delta}
$,
which contradicts the definition of $c_{F,1-\alpha}(h_{{G}})$.

\subsection{Proof of Theorem \ref{th:power}}
Write
$
 \HICM (\beta_{1})  =
 \HICM (\beta_{0})
-  2 \ip{\widehat{h}_{\beta_0}}{ \left(\beta_1 -
 \beta_0\right)'  \widehat{h}_{Y_2} }
+  \| \left(\beta_1 -  \beta_0\right)' \widehat{h}_{Y_2} \|^{2}
$,
where
\[
\widehat{h}_{Y_2}(s)   \equiv
n^{-1/2} \sum_{i=1}^{n}{
\left(b_0' \widehat{\Omega}(Z_i) b_0\right)^{-1/2} Y_{2i} \exp(i s'Z_{i})}
\, .
\]
\noindent
(i) From the results of Section \ref{sec:ucp},
\[
\left\| \tilde{c}_{n}^{-1} \left(  \widehat{h}_{Y_2}(s)  -
n^{-1/2} \sum_{i=1}^{n}{ \left(b_0' {\Omega}(Z_i) b_0\right)^{-1/2}
\Pi(Z_{i}) \exp(i s'Z_{i})} \right)  \right\|_{\infty}  \cvas 0
\]
as $\tilde{c}_{n} \rightarrow + \infty$  under Assumption
\ref{ass:pipower}-(i). Moreover
\[
\| \tilde{c}_{n}^{-1}
 n^{-1/2} \sum_{i=1}^{n}{ \left(b_0' {\Omega}(Z_i) b_0\right)^{-1/2}
\Pi(Z_{i}) \exp(i s'Z_{i})}
-
\E \left(b_0' {\Omega}(Z_i) b_0\right)^{-1/2} C(Z) \exp(i s'Z)
 \|_{\infty}
\cvas 0
\, ,
\]
uniformly in $P \in {\cal P}_{\beta_0}$.
Using  Lemma \ref{lem:boundicm}, we obtain
\begin{multline*}
\tilde{c}_{n}^{-2} \left( \HICM (\beta_{1}) - \HICM (\beta_{0}) \right)
 \cvas
\\
\left(\beta_1 -  \beta_0\right)'
\E \left[ \sigma^{-1}(Z_1) \sigma^{-1}(Z_2)
C(Z_{1})  C'(Z_{2}) \,
  w(Z_{1}-Z_{2})\right]
   \left(\beta_1 -  \beta_0\right)
\, .
\end{multline*}
By the arguments of \citet[Theorem 1]{Bierens82}, the above matrix is
positive semi-definite, and
\[
 a' \E \left( \sigma^{-1}(Z_1) \sigma^{-1}(Z_2) C(Z_{1}) C(Z_{2})  \,
 w(Z_{1}-Z_{2})\right) a \Rightarrow \quad a
= \mO \quad \mbox{or } \quad  \sigma^{-1}(Z)  C(Z) = \mO
\, .
\]
From Assumptions \ref{ass:variance}-(ii), \ref{ass:pipower}-(i), and
$\beta_1 \neq \beta_0$, the above limit is thus strictly positive.  Therefore
\begin{equation}
\lim_{\tilde{c}_{n} \rightarrow + \infty}
\sup_{P \in {\cal P}_{\beta_0}} \Pr \left[  \HICM (\beta_{1}) - \HICM
(\beta_{0})  > M \right] \rightarrow 1 \qquad \forall M >0
\, .
\label{unbound1}
\end{equation}

Assume now that the conclusion of Theorem \ref{th:power} does not
hold. Then there exists some $\delta>0$, an infinitely increasing
subsequence of sample sizes $n_{j}$, and a sequence of probability
measures $P_{n_{j}} \in {\cal P}_{\beta_{0}}$, with corresponding
sequences  $\Pi_{n_j}(\cdot)$ and
$\tilde{c}_{n_{j}}$, such that
\[
\Pr_{n_{j}} \left[ \HICM (\beta_{1}) < c_{1-\alpha} (h_{{G}}) \right]
> \delta
\qquad \forall n_{j}
\, .
\]
Then
\[
\Pr_{n_{j}} \left[ \HICM (\beta_{1}) - \HICM (\beta_{0})  < c_{1-\alpha} (h_{{G}})  - \HICM (\beta_{0})
\right] > \delta \qquad \forall n_{j}
\, .
\]
But  $\HICM (\beta_{0})$  is uniformly bounded in probability  by Lemma \ref{lem:boundicm}
and so is the critical value $c_{1-\alpha} (h_{{G}})$.
This contradicts (\ref{unbound1}).

(ii) Under Assumption \ref{ass:pipower}-(ii),
\[
 \HICM (\beta_{1n})  =  \HICM (\beta_{0})
-  2 \tilde{c}_n \ip{\widehat{h}_{\beta_0}}{ \delta'  \widehat{h}_{Y_2} } /\sqrt{n}
+  \tilde{c}_n^2  \| \delta'   \widehat{h}_{Y_2} \|^{2} / n
\, .
\]
We then proceed  as above to obtain that $\tilde{c}_{n}^{-2} \left(
\HICM (\beta_{1n}) - \HICM (\beta_{0})
\right)$ converges to a positive limit. The rest of the proof follows
similarly.

\subsection{Proof of Theorem \ref{th:powerspec}}
Consider
$
\HICM (\beta)  = \| \widehat{h}_{\beta} \|^{2}, \
\widehat{h}_{\beta}  = n^{-1/2} \sum_{i=1}^{n}{\left(b'\widehat{\Omega}(Z_i)
b\right)^{-1/2} Y_i'b \exp(is'Z_i)}
$.
Since $\beta$ belongs to the compact set ${\cal B}$,  use the same reasoning
as in Section \ref{sec:ucp} to show that the class of functions
$\left\{\left(b'{\Omega}(Z_i) b\right)^{-1/2} Y_i'b \exp(is'Z_i),
\beta \in {\cal B}, s \in M \right\}$ is BUEI for a
squared-integrable envelope. This yields
$\sup_{\beta \in {\cal B}} \left|
n^{-1} \HICM (\beta) -  H(\beta) \right|  =
o_p(1)$ uniformly in $P$, where
\[
H(\beta) = \int_{}^{}{ \left| \E \left[ (b' \Omega(Z) b)^{-1/2}
(y - Y'_{2}\beta) \exp(is'Z)  \right] \right|^2 \, d\mu(s) }
\, .
\]
Hence
$
n^{-1}   \HICM^*  - \min_{\beta \in {\cal B}} H(\beta)   = o_p(1)
$
uniformly in $P$.

Assume now that the conclusion of Theorem \ref{th:powerspec} does not
hold. Then there exists some $\delta>0$, an infinitely increasing
subsequence of sample sizes $n_{j}$ and a sequence of probability
measures $P_{n_{j}} \in {\cal P} \cap H_{1n}$ such that
$
\Pr_{n_{j}} \left[ \HICM^{*} < c_{1-\alpha} (h_{{G}}) \right]
> \delta
\, ,   \forall n_{j}
$.
Then
\[
\Pr_{n_{j}} \left[ n^{-1}  \HICM^*  - \min_{\beta \in {\cal B}} H(\beta)    <
n^{-1} c_{1-\alpha} (h_{{G}})  -  \min_{\beta \in {\cal B}} H(\beta)
\right] > \delta \ \forall n_{j}
\, .
\]
The critical value is bounded in probability, and under $H_{1n}$, $ n
\min_{\beta \in {\cal B}}
H(\beta)= \tilde{d}_n^{2}$. So whenever $\tilde{d}_n \rightarrow + \infty$, the
quantity $n^{-1} c_{1-\alpha} (h_{{G}})  -  \min_{\beta \in {\cal B}}
H(\beta)$ tends to $-\infty$, and we reach a contradiction.


\section{Additional empirical results}
\label{sec:appfig}

\begin{figure}[ht]
\begin{center}
\includegraphics[width = 6cm, height = 3cm]{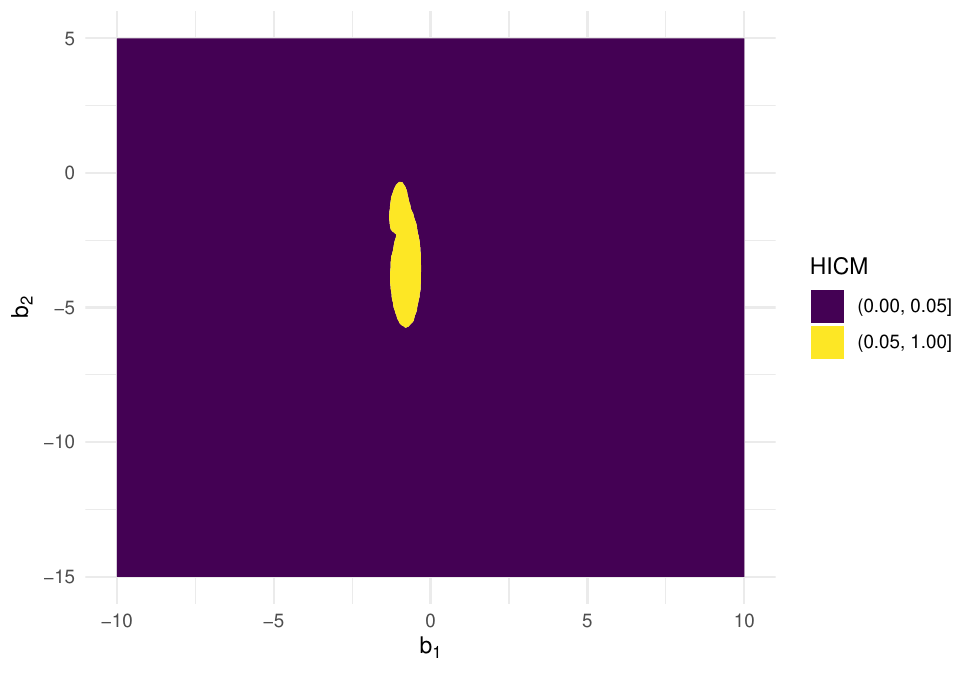}
\includegraphics[width = 6cm, height = 3cm]{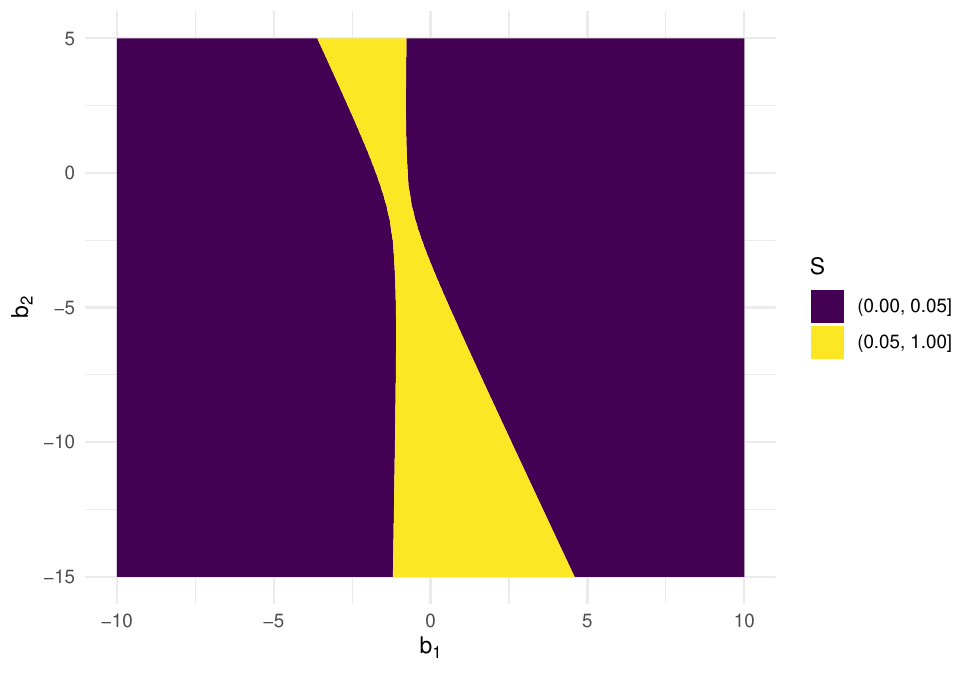}
\end{center}
\caption{95\% confidence regions for the coefficients of the
 population collapse obtained with HICM (left) and S (right) for
 Specification S2 using a  two-dimensional grid  over
 $[-10,10]\times[-15,5]$, with  201 evenly spread points on each interval.}
\label{graph: mexico NL S extended}
\end{figure}

\end{document}